\documentclass[11pt]{article}
\usepackage[margin=1in]{geometry}
\usepackage{fullpage}
\usepackage{tgtermes}
\usepackage[T1]{fontenc}
\usepackage[colorlinks,citecolor=blue,linkcolor=blue,urlcolor=black]{hyperref}
\usepackage{graphicx}
\usepackage{mathtools}
\usepackage{amsfonts,amsmath,amsthm,amssymb,dsfont}
\usepackage{thmtools}
\usepackage{xfrac,nicefrac}
\usepackage{mathdots}
\usepackage{bm,bbm}
\usepackage{url}
\usepackage{paralist}
\usepackage{enumerate}
\usepackage[normalem]{ulem}
\usepackage{xspace}
\xspaceaddexceptions{]\}}
\usepackage[capitalise]{cleveref}
\usepackage{comment}
\usepackage{tabu}
\usepackage{framed}
\usepackage{float,wrapfig}
\usepackage{tikz}
\usepackage[usenames,dvipsnames]{pstricks} 

\usepackage{enumitem}
\usepackage{adjustbox}
\usepackage{algorithm}
\usepackage{algpseudocode}

\algdef{SE}[DOWHILE]{Do}{doWhile}{\algorithmicdo}[1]{\algorithmicwhile\ #1}%
\algnewcommand{\algorithmicand}{\textbf{AND }}
\algnewcommand{\algorithmicor}{\textbf{OR }}
\algnewcommand{\algorithmicxor}{\textbf{XOR }}
\algnewcommand{\algorithmicnot}{\textbf{NOT }}
\algnewcommand{\OR}{\algorithmicor}
\algnewcommand{\AND}{\algorithmicand}
\algnewcommand{\XOR}{\algorithmicxor}
\algnewcommand{\NOT}{\algorithmicnot}
\algnewcommand{\var}{\texttt}

\algnewcommand{\algorithmicbreak}{\textbf{break}}
\algnewcommand{\Break}{\algorithmicbreak}

\usetikzlibrary{hobby}
\usetikzlibrary{decorations.markings}
\usetikzlibrary{quotes,angles}

\DeclarePairedDelimiter\abs{\lvert}{\rvert}%
\DeclarePairedDelimiter\norm{\lVert}{\rVert}%

\makeatletter
\let\oldabs\abs
\def\abs{\@ifstar{\oldabs}{\oldabs*}}
\let\oldnorm\norm
\def\norm{\@ifstar{\oldnorm}{\oldnorm*}}
\makeatother

\theoremstyle{plain}

\newtheorem{theorem}{Theorem}[section]
\newtheorem{lemma}[theorem]{Lemma}
\newtheorem{fact}[theorem]{Fact}

\theoremstyle{definition}

\def\ShowAuthNotes{1}
\ifnum\ShowAuthNotes=1
\newcommand{\authnote}[2]{\ \\ \textcolor{red}{\parbox{0.9\linewidth}{[{\footnotesize {\bf #1:} { {#2}}}]}}\newline}
\else
\newcommand{\authnote}[2]{}
\fi

\renewcommand{\epsilon}{\varepsilon}
\newcommand{\eps}{\varepsilon}

\renewcommand{\tilde}{\widetilde}

\newcommand{\union}{\cup}
\newcommand{\intersect}{\cap}
\newcommand{\floor}[1]{\left\lfloor #1 \right\rfloor}
\newcommand{\ceil}[1]{\left\lceil #1 \right\rceil}

\DeclareMathOperator*{\argmin}{arg\,min}

\newcommand{\cg}{\textsc{Congestion}\xspace}

\newcommand{\ver}{{V}\xspace}
\newcommand{\e}{{E}\xspace}

\title{Fully Dynamic All-Pairs Shortest Paths: Likely Optimal Worst-Case Update Time}

\author{  Xiao Mao \\  Stanford University \\  \texttt{xiaomao@stanford.edu} \\}

\date{}

\begin{document}
	\maketitle
	\thispagestyle{empty}
	\begin{abstract}
	The All-Pairs Shortest Paths (APSP) problem is one of the fundamental problems in theoretical computer science. It asks to compute the distance matrix of a given $n$-vertex graph. We revisit the classical problem of maintaining the distance matrix under a \emph{fully dynamic} setting undergoing vertex insertions and deletions with a fast \emph{worst-case} running time and efficient space usage.
	
	Although an algorithm with amortized update-time $\tilde O(n ^ 2)$\footnote{Throughout this paper, we use $\tilde O(f)$ to denote $O(f\cdot \text{poly} \log (f))$. } has been known for nearly two decades [Demetrescu and Italiano, STOC 2003], the current best algorithm for worst-case running time with efficient space usage runs is due to [Gutenberg and Wulff-Nilsen, SODA 2020], which improves the space usage of the previous algorithm due to [Abraham, Chechik, and Krinninger, SODA 2017] to $\tilde O(n ^ 2)$ but fails to improve their running time of $\tilde O(n ^ {2 + 2 / 3})$. It has been conjectured that no algorithm in $O(n ^ {2.5 - \epsilon})$ worst-case update time exists. For graphs without negative cycles, we meet this conjectured lower bound by introducing a Monte Carlo algorithm running in randomized $\tilde O(n ^ {2.5})$ time while keeping the $\tilde O(n ^ 2)$ space bound from the previous algorithm. Our breakthrough is made possible by the idea of ``hop-dominant shortest paths,'' which are shortest paths with a constraint on hops (number of vertices) that remain shortest after we relax the constraint by a constant factor.
	\end{abstract}
	\newpage

\section{Introduction}
The All-Pairs Shortest Paths (APSP) problem asks to compute the shortest distance between each pair of vertices on an $n$-vertex graph, called the \emph{distance matrix}. Despite being one of the most widely studied problems in theoretical computer science, on graphs with general weights the fastest algorithm still runs in $n ^ 3 / 2 ^ {\Omega(\sqrt{\log {n}})}$ time \cite{rrwapsp2014}, and thus it has been conjectured that an algorithm running in $O(n ^ {3 - \eps})$ time is impossible for any $\eps > 0$. However, under many \emph{dynamic} settings of the APSP problem where the graph undergoes modifications such as vertex/edge insertion/deletions, the distance matrix can be updated in a faster time than a cubic re-computation from scratch. 

\paragraph{The Problem.} In this paper, we study the fully dynamic setting under vertex insertion/deletions with a faster worst-case time. Formally, we are given a graph $G := \langle \ver, \e \rangle$ with $\abs{\ver} = n$ that undergoes vertex insertions and deletions where each vertex is inserted/deleted together with its incident edges. Our goal is to minimize the \emph{update time}, which is the time needed to refresh the distance matrix. Specifically, since APSP without any updates already takes $\tilde O(n ^ 3)$ time, we allow a preprocessing procedure that takes place before the first update whose running time is not part of the update time since it does not ``refresh'' the distance matrix. We measure the worst-case update time which is opposed to amortized update time. In the randomized setting, the worst-case update time is a time bound under which \emph{every} update can finish with the correct answer with a probability of $1 - n ^ {-c}$ for some constant $c > 0$.

We point out that in the fully dynamic setting, vertex updates are always more general than edge updates since any edge update can be easily simulated by a deletion and a re-insertion of one of the two incident vertices.

\paragraph{Previous Work.}Algorithms for dynamic All-Pairs Shortest Path Problem have been known since the 1960s, the earliest partially dynamic algorithm was the well-known Floyd-Warshall Algorithm \cite{floyd,warshall} from as early as 1962\footnote{The algorithm is essentially the same as earlier algorithms from 1959 by Bernard Roy \cite{frenchpaperfrom1959}.} that can be easily extended to handle vertex insertion for $O(n ^ 2)$ time per update given the distance matrix of the current $n$-vertex graph. It was not until the end of that century that the first fully dynamic algorithm was given by King \cite{king99}, with $\tilde O(n ^ {2.5}\sqrt{W \log n})$ amortized update time per \emph{edge} insertion/deletion where $W$ is the largest edge weight. Their algorithm was based on a classic data structure for decremental Single-Source Shortest Paths by Even and Shiloach \cite{Even1981AnOE}. Later, King and Thorup \cite{king01} improved the space bound from $\tilde O(n ^ 3)$ to $\tilde O(n ^ {2.5}\sqrt{W})$. The follow-up work by Demetrescu and Italiano \cite{demetrescu02, demetrescu06} generalized the result to real edge weights from a set of size $S$. In 2004, a breakthrough by Demetrescu and Italiano \cite{demetrescu04} gave an algorithm with $\tilde O(n ^ 2)$ amortized update time per \emph{vertex} insertion/deletion using $\tilde O(n ^ 3)$ space. Thorup \cite{amortized} simplified their approach, shaved some logarithmic factors, and extended it to handle negative cycles. Based on his approach, he developed the first fully dynamic algorithm with a better \emph{worst-case} update time for vertex insertion/deletions than recomputation from scratch with $\tilde O(n ^ {2 + 3/4})$ time per update, using a space super-cubic in $n$. In 2017, Abraham, Chechik and Krinninger \cite{Abraham2017FullyDA} improved the space to $O(n ^ 3)$ and designed a randomized algorithm with $\tilde O(n ^ {2 + 2 / 3})$ worst-case update time per vertex update against an adaptive adversary. The current state-of-the-art algorithms are by Gutenberg and Wulff-Nilsen in 2020 \cite{previous}. They achieved $\tilde O(n ^ {2 + 4 / 7})$ deterministic update time but failed to improve the randomized update time and instead only managed to improve space usage to $\tilde O(n ^ 2)$. Interestingly, their algorithm is a Las Vegas algorithm while both our algorithm and Abraham, Chechik and Krinninger's algorithm are Monte Carlo algorithms\footnote{The difference between Monte Carlo algorithms and Las Vegas algorithms is that although both types of algorithms may have randomized running times, the former may be incorrect with a small probability while the latter are always correct.}. For deterministic algorithms, Chechik and Zhang recently proposed a deterministic algorithm that runs in $\tilde O(n ^ {2 + 41 / 61})$ time \cite{doi:10.1137/1.9781611977554.ch4}. The history of the fully dynamic APSP problem is further summarized in Table \ref{table:apsp}. 
\begin{table}[h]
\begin{center}
\begin{tabular}{ |c|c|c|c|c|c|c|c| }
 \hline
     Time & Space & VU & WC & Type & $\textrm{C} ^ {-}$ & Authors & Year\\
 \hline
 $\tilde O(n ^ {2.5}\sqrt{W})$    & $\tilde O(n ^ 3)$     & N    & N & Det & N  & King \cite{king99}  & 1999 \\
 $\tilde O(n ^ {2.5}\sqrt{W})$  & $\tilde O(n\sqrt{W})$ & N    & N  & Det & N  & King and Thorup \cite{king01}  & 2001 \\
 $\tilde O(n ^ {2.5}\sqrt{S})$  & $\tilde O(n\sqrt{S})$ & N     & N  & Det & Y  & Demetrescu and Italiano \cite{demetrescu02}  & 2002 \\
 $\tilde O(n ^ 2)$              & $\tilde O(n ^ 3)$     & Y  & N  & Det & N  & Demetrescu and Italiano \cite{demetrescu04}  & 2004 \\
 $\tilde O(n ^ 2)$              & $\tilde O(n ^ 3)$    & Y  & N  & Det & Y & Thorup \cite{amortized}  & 2004 \\
 $\tilde O(n ^ {2 + 3 / 4})$    & $\omega (n ^ 3)$      & Y  & Y & Det & N  & Thorup \cite{thorup05}   & 2005 \\
 $\tilde O(n ^ {2 + 3 / 4})$    & $O(n ^ 3)$            & Y  & Y & Det & N  & Abraham, Chechik and Krinninger \cite{Abraham2017FullyDA}   & 2017 \\
 $\tilde O(n ^ {2 + 2 / 3})$    & $O(n ^ 3)$            & Y  & Y & MC  & N  & Abraham, Chechik and Krinninger \cite{Abraham2017FullyDA}   & 2017 \\
 $\tilde O(n ^ {2 + 5 / 7})$    & $\tilde O(n ^ 3)$     & Y  & Y & Det & Y & Gutenberg and Wulff-Nilsen \cite{previous}  & 2020 \\
 $\tilde O(n ^ {2 + 2 / 3})$    & $\tilde O(n ^ 2)$     & Y  & Y & LV & Y & Gutenberg and Wulff-Nilsen \cite{previous}  & 2020 \\
 $\tilde O(n ^ {2 + 41 / 61})$    & $\tilde O(n ^ {3+5/61})$     & Y  & Y & Det & Y & Chechik and Zhang \cite{doi:10.1137/1.9781611977554.ch4}  & 2023 \\
 $\tilde O(n ^ {2.5})$          & $O(n ^ 2)$     & Y  & Y & MC & N & \textbf{This paper}  &  \\
 \hline
\end{tabular}
\end{center}
\caption{Fully dynamic APSP algorithms, where ``VU'' stands for ``vertex updates'' (as opposed to edge updates), ``WC'' stands for ``worst-case'' (as opposed to amortized), ``Det/MC/LV'' stands for ``deterministic/Las Vegas/Monte Carlo,'' and ``$\textrm{C} ^ {-}$'' stands for ``(allow) negative cycles.''} \label{table:apsp}
\end{table}

It has been observed by Abraham, Chechik and Krinninger that $\tilde O(n ^ {2.5})$ seems to be a natural barrier for the worst-case update time \cite{Abraham2017FullyDA}. It is a natural question whether there is an algorithm that meets this natural barrier and whether there are impossibility results. Our result gives an algorithm that meets this natural barrier while preserving the quadratic space usage.

\paragraph{Our Result.} Suppose that the graph does not contain a negative cycle at any time. Our main result is a randomized algorithm with a running time that meets the natural barrier of $\tilde O(n ^ {2.5})$.
 
\begin{theorem} \label{theo:main}
    Let $G$ be an $n$-vertex edge-weighted directed graph undergoing vertex insertions and deletions with no negative cycles at any time. There exists a Monte-Carlo data structure that can maintain distances in $G$ between all pairs of vertices in randomized worst-case update time $\tilde O(n ^ {2.5})$ with high probability\footnote{In this paper, the term ``high probability'' implies a probability of at least $1 - n ^ {-c}$ for any constant $c > 0$.}, using $\tilde O(n ^ 2)$ space, against an adaptive adversary.
\end{theorem}

\subsection*{Related Results}
\paragraph{Partially Dynamic Algorithms.} In the partially dynamic model only one type of update is allowed. The \emph{incremental} model is restricted to insertions and the \emph{decremental} model is restricted to deletions. The amortized update time bounds for partially dynamic algorithms often depend on the size of the graph rather than the number of updates. Thus it is often convenient to report the \emph{total update time}, which is the sum of the individual update times. In the incremental setting, the earliest partially dynamic algorithm was the well-known Floyd-Warshall Algorithm \cite{floyd,warshall} from as early as 1962 that can be easily extended to handle vertex insertion for $O(n ^ 2)$ time per update, which means $O(n ^ 3)$ total update time for $n$-node graphs. Apart from this, there is an algorithm by Ausiello, Italiano, Spaccamela and Nanni \cite{10.5555/320176.320178} for edge insertions that has a total update time of $O(n ^ 3W\log{(nW)})$ where $W$ is the largest edge weight. In the decremental setting, the fastest algorithm for vertex deletions runs in $O(n ^ 3\log n)$ total update time \cite{demetrescu04}, while the fastest algorithm for edge deletions runs in randomized $O(n ^ 3\log ^ 2n)$ total update time for unweighted graphs \cite{Baswana2003MaintainingAA}, and $O(n ^ 3S\log ^ 3n)$ total update time for weighted graphs \cite{demetrescu06} if edge weights are chosen from a set of size $S$.

\paragraph{Approximation Algorithms.} The problem of dynamic All-Pairs Shortest Paths in the approximate setting has been studied in various dynamic graph settings \cite{baswana02,demetrescu04,roddity04,bernstein09,zwick04,abraham13,bernstein16,henzinger16,chechik18,8948617,Karczmarz2019ReliableHF, doi:10.1137/1.9781611975994.154, doi:10.1137/1.9781611976014.15, chuzhoy21, evald_et_al:LIPIcs.ICALP.2021.64, bernstein2022deterministic}. In the setting of $(1 + \epsilon)$-approximate shortest paths on an $n$-vertex, $m$-edge graph, although many algorithms achieve amortized update time $\tilde O(m / \epsilon)$, for the worst-case running time, the only of these algorithms that give a better guarantee than the trivial $\tilde O(mn / \epsilon)$ bound is the one in \cite{8948617} with a running time of $\tilde O(n ^ {2.045} / \epsilon ^ 2)$ on digraphs with positive edge weights based on fast matrix multiplication. There is also a work studying the related problem of fully dynamic minimum weight cycle under vertex updates in the approximate settings \cite{Karczmarz2021FullyDA}.

\paragraph{Special Graphs.}
There are also many results on fully dynamic algorithms for planar \cite{kleinplanar98, henzinger97, rao06, abraham12} and undirected graphs \cite{zwick04, bernstein09, bernstein16, abraham14, henzinger16, chuzhoy21}. For sparse graphs. recently, Karczmarz and Sankowski developed a fully dynamic algorithm for vertex updates on a $n$-vertex $m$-edge graph with real weights with $\tilde O(mn ^ {4 / 5})$ worse case update time and $\tilde O(n ^ {4 / 5})$ time for a single $s$-$t$ distance query \cite{karczmarz_et_al:LIPIcs.ICALP.2023.84}.

\section{Techniques Overview}
\subsection{The Current Framework and a Likely Lower Bound}
Consider a much-simplified version of our problem: suppose that we know all the upcoming deletions, but not the insertions. It turns out that if we are given an oracle that tells us on which vertices the next $2n ^ {0.5}$ deletions are, there is an approach that achieves $\tilde O(n ^ {2.5})$ worst-case update time. 

The approach is as follows. After every $n ^ {0.5}$ updates, we take a snapshot of the current graph and compute the distance matrix of the graph resulting from doing all $2n ^ {0.5}$ future deletions from the snapshot. This takes $\tilde O(n ^ 3)$ time in total. However, a standard de-amortization technique is to divide the running time into $n ^ {0.5}$ chunks and distribute them equally to the next $n ^ {0.5}$ updates. This way each update still takes $\tilde O(n ^ {2.5})$ worst-case update time. At every update, suppose the computation of the penultimate distance matrix was started at $D$ updates ago. Then we will have $n ^ {0.5} \le D \le 2n ^ {0.5}$. This distance matrix has therefore been fully computed and it is the distance matrix of some graph that is a sub-graph of the current graph with at most $3n ^ {0.5}$ missing vertices from extra deletions and missing insertions. We can reinsert these missing vertices using the Floyd-Warshall Algorithm in $O(n ^ 2)$ time per insertion, which still takes only $O(n ^ {2.5})$ time in total.

Despite its reliance on an oracle, this simple approach's ideas of preprocessing, batch deletion, and batch insertion serve as the basis for all of the current line of algorithms achieving a worst-case running time better than cubic, including ours. Since Thorup's first data structure achieving a worst-case update time of $\tilde O(n ^ {2 + 3 / 4})$, we have essentially been using an algorithm consisting of two steps:
\begin{itemize}
    \item Preprocessing: take a snapshot of the current graph, and build a data structure so that we can deal with at most $\Delta$ upcoming deletions efficiently.
    \item Batch deletion: return the distance matrix after removing at most $\Delta$ vertices from the last snapshot.
\end{itemize}
As before, the running time for preprocessing can be split into chunks to guarantee worst-case update time. At each update, we run the batch deletion to deal with all the deletions since the last snapshot, and then we re-insert the missing vertices due to extra deletions and missing insertions using the Floyd-Warshall algorithm.

It is reasonable to believe that $\tilde O(n ^ {2.5})$ is the best running time one can achieve within such a framework. Moreover, it seems pretty hard for any algorithm to get around this $O(n ^ {2.5})$ barrier as an algorithm that breaks this barrier seems elusive even in the presence of an oracle with information on all future deletions, which makes it reasonable to conjecture that no data structure can achieve a worst-case update time with running time $O(n ^ {2.5 - \eps})$ for any $\eps > 0$. Since the current best algorithm has a worst-case update time of $\tilde O(n ^ {2 + 2 / 3})$. It is a very important open problem whether there exists an algorithm that matches the running time barrier of $\tilde O(n ^ {2.5})$, to which our paper gives a positive answer. To introduce our new techniques for achieving such a running time, we first give a brief description of Gutenberg and Wulff-Nilsen's data structure.

In the preprocessing step, their algorithm finds a small set $C$ of ``congested vertices" and precomputes a set of paths. They do it in a way such that if we remove some small set $D$ of vertices from our graph in the batch deletion step, we have a way to efficiently concatenate some surviving precomputed paths so that we can recover pair-wise paths in $G \backslash D$ that are at most as long as the corresponding shortest paths in $G \backslash (C \union D)$. For each update, we can call the batch deletion step which recovers these pair-wise paths, and then if we reinsert the vertices in $C \backslash D$, we recover path-wise shortest paths in $G \backslash D$. In order to be able to efficiently recover new paths using precomputed paths, they make use of the fact that finding the best concatenation of paths with more hops (number of vertices) can be sped up with hitting sets. To harmonize with this idea, in the processing step, they pre-compute pair-wise \emph{hop-restricted} paths with different scales of hop restrictions with congestion incurred by a path defined to be proportional to the time needed to recover a path with the same hop. Reinserting the vertices in $C \backslash D$ turns out to be inefficient. A better strategy is to maintain pairwise shortest paths that pass through at least one center in $C$ subject to batch deletion of $D$. To achieve maximum efficiency, they extended this idea to multilayer. Every layer computes the set of congested vertices and delegates the job of maintaining the paths through the congested vertices to the layer below. 

\subsection{Breakthrough Using Hop-Dominant Paths}
The barrier in both Abraham, Chechik and Krinninger's data structure and Gutenberg and Wulff-Nilsen's data structure comes from the following \emph{hop-restricted shortest path} problem: given a source vertex $s$ in the graph and a \emph{hop limit} $h$, we want to know for each other vertex $t$ the shortest $s$-$t$ path that visits at most $h$ edges (called $h$-hop shortest paths). Unfortunately, for this problem, we are yet to do better than the classic solution using the Bellman-Ford algorithm in $O(n ^ 2h)$ time. In our breakthrough, we go around this barrier by only computing \emph{hop-dominant paths}, which are $h$-hop shortest paths that are also $2h$-hop shortest (i.e. no path with a hop in $[h + 1, 2h]$ is shorter). Our core idea is that for the same set of congested vertices, we can recover all-pair shortest paths subject to vertex deletions by efficiently concatenating these hop-dominant paths as well. Fortunately, it turns out that hop-dominant paths can be computed in $O(n ^ 2)$ time (\cref{sec:ssahdp}), which speeds up our computation.

However, simply replacing the hop-restricted shortest paths with hop-dominant paths does not readily give us an improvement. Such replacement does give us a set of congested vertices and a set of hop-dominant paths, but it seems impossible to implement the batch deletion procedure to obtain the same running time guarantee (\cref{sec:fail}). We deal with this issue using two steps. In the first step (\cref{sec:step1}), we show that we can obtain the set of congested vertices using hop-restricted paths in an approximate way that does not require us to actually compute all hop-restricted paths. We can show that this gives us an efficient batch deletion procedure that requires an oracle. In the second step (\cref{sec:step1}), we show how we can efficiently concatenate hop-dominant paths without using the oracle.

\paragraph{Randomization on a Sequential Sub-Routine.} In the previous algorithm, they obtain hop-restricted paths as well as the congested vertices altogether using a sequential sub-routine. This sub-routine computes single-source hop-restricted paths from each vertex in some order and modifies the graph on the fly based on newly congested vertices. The technique we employ to speed up the previous sub-routine is by applying randomization. Instead of going through the vertices one by one, we randomly select some vertices, compute only the single-source hop-restricted paths from these vertices, and magnify the modification on the graph so that it approximates the original deterministic procedure (\cref{sec:randomization}). Our argument here is that even if we skip the majority of the vertices, our randomized sub-routine is still a good approximation of the original deterministic sub-routine.

\paragraph{Construction for Efficient Concatenations.} We show that if we are given an oracle that tells us the final length and hop of each shortest $s$-$t$ path, we can efficiently concatenate hop-dominant paths to ``witness'' these shortest paths. To get rid of this oracle, we categorize the recovered paths into two types. The first type is those that are concatenations of exactly two hop-dominant paths, and the second type is those that are concatenations of more than two hop-dominant paths. We deal with the first type using a copy of Gutenberg and Wiff-Nilsen's data structure with some modification (\cref{sec:typei}), and the second type with a Dijkstra-like sub-routine (\cref{sec:typeii}).

\section{Preliminaries}
In this paper, we will largely inherit the notations used in \cite{previous}. We denote by $G := \langle \ver, \e \rangle$ the input digraph where the edges are associated with a weight function $w: \e \rightarrow \mathbb{R}$. For the input graph $G$, we use $n$ to denote $\abs{\ver}$ and $m$ to denote $\abs{\e}$. For a subset $\ver ^ {\prime} \subset \ver$, we denote by $G[\ver ^ {\prime}]$ the induced subgraph of $\ver ^ {\prime}$. For $D \subset \ver$ we write $G \backslash D$ as a shorthand for $G[\ver \backslash D]$. We let $\overleftarrow{G}$ be the digraph identical to $G$ but with all edge directions reversed. For two vertices $(u, v) \in \ver ^ 2$, let $\langle u, v \rangle$ denote the edge from $u$ to $v$, and let $w(u, v)$ denote the weight of $\langle u, v \rangle$.

For a path $p$ from $s$ to $t$ containing $h$ edges and $h + 1$ vertices, we let $p[i]$ be the $i$-th vertex on $p$ (0-based) where $p[0] = s$ and $p[h] = t$. We can describe $p$ by $\langle p[0], p[1], \cdots, p[h] \rangle$. For $0 \le l \le r \le h$ we write $p[l, r]$ to denote the segment of $p$ from the $l$-th to the $r$-th vertex. $h$ is called the \emph{hop} of $p$ and is denoted by $\abs{p}$. The \emph{weight} or \emph{length} of $p$, $w(p)$, is defined as the sum of the weights of the edges on $p$. We say a path is shorter or longer than another by comparing the lengths (not the hops) of the two paths. We use $\perp$ to denote an empty path, whose weight is defined to be $\infty$. Given two paths $p_1$ and $p_2$ such that $p_1[\abs{p_1}] = p_2[0]$, we denote by $p_1 \circ p_2$ the \emph{concatenated path} $\langle p_1[0], p_1[1], \cdots, p_1[\abs{p_1}] = p_2[0], p_2[1], \cdots, p_2[\abs{p_2}] \rangle$. Specifically, for any path $p$ we define $p\ \circ \perp := \perp \circ\ p := \perp$. We use $v \in p$ to denote when a vertex $v$ is on the path, or equivalently, $v = p[k]$ for some $k \in [0, h]$, and we let the \emph{prefix} $p[<v]$ be the segment of the path before $v$, or equivalently the path $p[0, k]$, and the \emph{suffix} $p[>v]$ be the segment after $v$, or equivalently $p[k, h]$. We let $\overleftarrow{p}$ be the reversed path $\langle p[h], p[h - 1], \cdots, p[0] \rangle$ in $\overleftarrow{G}$.

For a path $p$ from $s$ to $t$. We say $p$ is the \emph{shortest path} in $G$ if the weight of $p$ is minimized over all $s$-$t$ paths. $p$ is an $h$\emph{-hop-restricted path} (sometimes shortened to $h$-hop path) if its hop does not exceed $h$. $p$ is the $h$\emph{-hop-shortest path} if the weight of $p$ is the minimized over all $h$-hop-restricted $s$-$t$ paths. For a set of vertices $C$, we say a path is through $C$ if the path visits at least one vertex in $C$.

For our convenience, from now on we assume the size of the graph $n$ is at least a large constant.

\paragraph{Uniqueness of Shortest Paths.}In this paper we assume that for every $h$, the $h$-hop shortest paths between a pair of vertices are unique. We shall describe how this is achieved in \cref{sec:breaktie}.

\subsection*{Ingredients}
We will use the following well-known result for inserting vertices:
\begin{lemma} [See Floyd-Warshall Algorithm \cite{floyd, warshall}] \label{lemma:floyd}
    Given a graph $G \backslash C$, where $C \subset V$ of size $\Delta$, and given its distance matrix $D$ in $G \backslash C$. There is an algorithm that computes the distance matrix of $G$ in $O(\Delta n ^ 2)$ time and $O(n ^ 2)$ space.
\end{lemma}

Given a hop limit $h$, the single source $h$-hop shortest paths can be computed in $O(n ^ 2h)$ time using the famous Bellman–Ford algorithm. We can modify the algorithm so that it computes the shortest paths through a set of centers:
\begin{lemma}
    Given a graph $G := \langle \ver, \e \rangle$ of $n$ vertices. Let $C \subset \ver$ be a set of centers. For a starting vertex $s$ and a hop limit $h$, there is a procedure $\textsc{BellmanFordThroughCenter}(s, G, C, h)$ that computes for each destination $t$ the shortest among paths between $s$ and $t$ through $C$ with a hop at most $h$ in $\tilde O(n ^ 2h)$ time and $\tilde O(n ^ 2)$ space. 
\end{lemma}
\begin{proof}
    In order to force paths through a center in $C$, we split each vertex $v$ into $(v, 0)$ and $(v, 1)$ where the second parameter is $1$ if and only if we visited at least one vertex in $C$. We copy each edge from $u$ to $v$ into two edges from $(u, 0)$ to $(v, 0)$ and $(u, 1)$ to $(v, 1)$, respectively. If $v \in C$, we also add an edge from $(u, 0)$ to $(v, 1)$. To force the paths to pass through $C$, we can simply run the algorithm on this new graph from $(s, 0)$ and for each destination $t$, we take the final distance to $(t, 1)$.
\end{proof}
The output paths can easily be stored in $O(nh)$ space. In Gutenberg and Wulff-Nilsen's paper, they introduced a data structure that can store all the output paths in $\tilde O(n)$ space, but this will be unnecessary for our algorithm.

We will use a well-known fact that good hitting sets can be obtained by random sampling. This technique was first used in the context of shortest paths by Ullman and Yannakakis \cite{ullman}. The fact can be formalized in the following lemma:
\begin{lemma} [See \cite{ullman}]\label{lemma:hsrd}
    Let $N_1, N_2, \cdots, N_n \subset U$ be a collection of subsets of $U$, with $u := \abs{U}$ and $\abs{N_i} \ge s$ for all $i \in [1, n]$. For any constant $a \ge 1$, let $A$ be a subset of $U$ that was obtained by choosing each element of $U$ independently with probability $p := \min(x / s, 1)$ where $x := a \ln{(un)} + 1$. Then with a probability of at least $1 - 1 / n ^ a$, the following two properties hold:
    \begin{itemize}
        \item $N_i \cap A \ne \emptyset$ for all $i$;
        \item $\abs{A} \le 3xu / s = O(u \ln{(un)} / s) = \tilde O(u / s)$.
    \end{itemize}
\end{lemma}

We will also use the following lemma that efficiently computes the shortest paths that have a hop no less than a threshold $h$.
\begin{lemma} [See \cite{Abraham2017FullyDA}] \label{lemma:long}
    Given a graph $G := \langle \ver, \e \rangle$ of $n$ vertices with non-negative edge weights, for any $h > 0$, there exists a randomized procedure $\textsc{RandGetShortestPaths}(G, h)$ running in $\tilde O(n ^ 3 / h)$ time using $\tilde O(n)$ space that computes a distance matrix $A = a_{s, t}$ such that for every $(s, t) \in \ver ^ 2$:
    \begin{itemize}
        \item If we let ${\pi}_{s, t}$ be the shortest $s$-$t$ path, then $w({\pi}_{s, t}) \le A_{s, t}$.
        \item If we let ${\Pi}_{s, t}$ be the shortest $s$-$t$ path with a hop of at least $h$, then $A_{s, t} \le w({\Pi}_{s, t})$ with high probability.
    \end{itemize}
\end{lemma}
To give an intuition for \cref{lemma:long}, observe that from \cref{lemma:hsrd} we can sample a hitting set of size $\tilde O(n / h)$ for all shortest paths with a hop of at least $h$. We can compute the shortest paths to and from vertices in this set by running Dijkstra's algorithm from each vertex in this set in $G$ as well as $\overleftarrow{G}$, and then combine the answers trivially. It is easy to see that the computed distance matrix satisfies the stated conditions. 

Finally, we will also use an online version of the Chernoff bound:
\begin{lemma} [Theorem 2.2 of \cite{Chekuri2018RandomizedMF}] \label{lemma:onlinechernoff}
    Let $X_1, \cdots X_n, Y_1, \cdots Y_n \in [0, 1]$ be random variables and let $\eps \in [0, 1 / 2)$ be a sufficiently small constant. If $\mathbf{E}[X_i \mid X_1, \cdots X_{i - 1}, Y_1 \cdots Y_i] \ge Y_i$ for each $i$, then for any $\delta > 0$, $$\mathrm{P}\left(\sum\limits_{i = 1} ^ {n}{X_i} \le (1 - \eps)\sum\limits_{i = 1} ^ {n}{Y_i} - \delta\right) \le (1 - \eps) ^ {\delta}.$$
\end{lemma}

\subsection{Path Operations} \label{sec:po}
For now, we assume the following operations on paths are possible:
\begin{itemize}
    \item Given a path $p$, we can retrieve the list of vertices $\{p[0], p[1], p[2], \cdots, p[\abs{p}]\}$ in $O(\abs{p})$ time.
    \item Given two paths $u$ and $v$, we can concatenate $u$ and $v$ in $O(1)$ time.
    \item Both $w(p)$ and $\abs{p}$ can be found in $O(1)$ time.
\end{itemize}

After describing our algorithm, we shall discuss how all these path operations are possible within $\tilde O(n ^ 2)$ space in \cref{sec:space}.

\section{A Framework Based on Gutenberg and Wulff-Nilsen's Approach}
Gutenberg and Wulff-Nilsen's approach was based on an earlier fundamental framework that reduces the fully dynamic problem into a decremental one, which states the following reduction:
\begin{lemma} [See \cite{henzinger01,thorup05,Abraham2017FullyDA, previous}] \label{lemma:framework}
    Given a data structure on $G$ that supports a batch deletion of up to $\Delta$ vertices $D \subset \ver$ from $G$ such that the data structure computes for each $(s, t) \in (\ver \backslash D) ^ 2$ the shortest path ${\pi}_{s, t}$ in $G \backslash D$ with the preprocessing time begin $t_{\textrm{pre}}$, the batch deletion worst-case time being $t_{\textrm{del}}$, and the space usage being $M$, there exists a fully dynamic APSP algorithm with $O(t_{\textrm{pre}} / \Delta + t_{\textrm{del}} + \Delta n ^ 2)$ worst-case update time and $O(M)$ space.
\end{lemma}
%To give an intuition for \cref{lemma:framework}, we let $D$ be the set of vertex deletions accumulated from the updates, and $I$ be the set of vertex insertions. We use the data structure to compute the shortest paths in $G \backslash D$ with time $t_{\textrm{del}}$ and insert $I$ back to the graph in $O(I n ^ 2)$ time using \cref{lemma:floyd}. We preprocessing the data structure whenever $\abs{D} + \abs{I} \ge \Delta$ (after at least $\Delta$ updates) and clear $D$ and $I$. Using standard deamortization techniques, we can split the preprocessing time into small chunks that are processed at individual updates. This guarantees a worst-case running time of $O(t_{\textrm{pre}} / \Delta)$ per update.

From now on, we can suppose that there are no insertions. The framework by Gutenberg and Wulff-Nilsen in \cite{previous} is an extension of the fundamental framework. We create $L = O(\log n)$ layers indexed increasingly from bottom to top. The layer below layer $i$ is layer $i - 1$ and the layer above layer $i$ is layer $i + 1$. For each layer $i$, a subset of center vertices $C_i$ is given as a parameter. For the top layer $L$, $C_L = \ver$. For each layer $i$, we employ a data structure $S_i$ that supports a batch deletion of up to $2 ^ i$ vertices $D_i \subset V$. Each layer $i$ is preprocessed exactly per $2 ^ i$ updates, and the set $D_i$ is the set of vertices deleted since layer $i$ is last preprocessed, so $\abs{D_i} \le 2 ^ i$ and the batch deletion is always within its limitation. All layers are preprocessed initially. Therefore, every time a layer is preprocessed, all layers below it are preprocessed. When multiple layers are preprocessed, their corresponding preprocessing steps are invoked from top to bottom (i.e. in decreasing order of indices). The data structure for layer $i$ computes for each $(s, t) \in (V \backslash D_i) ^ 2$ some surviving $s$-$t$ path that is at most as long as the shortest $s$-$t$ path in $G \backslash C_{i - 1}$ through $C_i$ (Suppose $C_{-1} = \emptyset$). From a telescoping argument, if we combine the results from layer $i$ to layer $0$, we find the shortest paths through $C_i$. Since $C_L = \ver$, combining all layers gives us the globally shortest paths.

In \cite{previous}, based on standard de-amortization techniques, Gutenberg and Wulff-Nilsen implicitly showed the following lemma:
\begin{lemma} \label{lemma:frameworkdetailed}
    If for layer $i$, the preprocessing step runs in $\tilde O(T 2 ^ i)$ time, the batch deletion step runs in $\tilde O(T)$ time, and either step uses $\tilde O(M)$ space, and the number of layers $L$ is such that $2 ^ L$ is at most $\tilde O(T / n ^ 2)$, then we have a data structure for fully-dynamic APSP under vertex operations with a worst-case update time of $\tilde O(T)$ and a space usage of $\tilde O(M)$
\end{lemma}
%To give an intuition for \cref{lemma:frameworkdetailed}, first note that since layer $i$ is preprocessed every $2 ^ i$ updates, we can apply the deamortization technique to each layer to make their worst-case running time $\tilde O(T)$ per layer. Since there are only $\tilde O(1)$ layers, this gives us a worst-case running time of $\tilde O(T)$ for the incremental case. To handle an update, for each layer, we let $D_i$ be the set of vertices deleted since the last time layer $i$ was preprocessed, and call the batch deletion procedures. Since there are $\tilde O(1)$ layers and the batch deletion procedure takes $\tilde O(T)$ time for each layer, the total time for batch deletion is $\tilde O(T)$. The space usage is $\tilde O(M)$ since no extra space overhead is introduced.

In the previous algorithm by Gutenberg and Wulff-Nilsen from \cite{previous}, the preprocessing step for layer $i$ runs in $\tilde O(n ^ {8 / 3} \times 2 ^ i)$ time, the batch deletion step runs in $\tilde O(n ^ {8 / 3})$ time, either step uses $\tilde O(n ^ 2)$ space and the number of layer $L$ is such that $2 ^ L = \tilde O(n ^ {2 / 3})$\footnote{In their paper the data structure only contains the top $O(n ^ {1 / 3})$ layers but including the extra layers does not affect the overall running time up to a logarithmic factor.}. By applying \cref{lemma:frameworkdetailed}, this gives a running time of $\tilde O(n ^ {2 + 2 / 3})$ with a space usage of $\tilde O(n ^ 2)$. 

Our final algorithm will be such that:
\begin{itemize}
    \item For layer $i$, the preprocessing step runs in $\tilde O(n ^ {2.5} \times 2 ^ i)$ time,
    \item For every layer, the batch deletion step runs in $\tilde O(n ^ {2.5})$ time,
    \item For every layer, either step uses $\tilde O(n ^ 2)$ space, and
    \item Let $K$ be the largest power of $2$ that is not more than $n ^ {1 / 2}$, we have $2 ^ L = K$. 
\end{itemize}
By applying \cref{lemma:frameworkdetailed}, we get a running time of $\tilde O(n ^ {2.5})$ with a space usage of $\tilde O(n ^ 2)$, proving \cref{theo:main}.

To help fully introduce our data structure later, we will first introduce our basic data structure based on Gutenberg and Wulff-Nilsen's data structure with a lackluster preprocessing time of $\tilde O(n ^ {3.5})$ while achieving the goal of a batch deletion time of $\tilde O(n ^ {2.5})$ and an $L$ where $2 ^ L = \tilde O(n ^ {1 / 2})$.

Finally, we show how we can enforce non-negative edge weights using a transformation previously used in \cite{Abraham2017FullyDA}, assuming that the graph contains no negative cycles at any time.
\begin{lemma} [``Johnson transformation''\cite{Johnson1977EfficientAF}]
    Given a graph $G$ without negative cycles, we can re-weight its edges with non-negative values to obtain a graph $\tilde{G}$ such that for any $(s, t) \in \ver ^ 2$ and any two $s$-$t$ paths $(p_0, p_1) \in G$ (and therefore also in $\tilde{G}$), if $\pi_0$ is shorter than $\pi_1$ is $G$ then $\pi_0$ is also shorter than $\pi_1$ in $\tilde{G}$ and vice versa. Such $\tilde{G}$ can be computed in $O(n ^ 3)$ time.
\end{lemma}
\begin{proof}
    We use a \emph{potential function} $p: \ver \rightarrow \mathbb{R}$ such that for every edge $u, v \in \e$ with weight $w(u, v)$, the modified weight $w ^ {\prime}(u, v) = w(u, v) + p(u) - p(v) \ge 0$. For any $s$-$t$ path $p$ of length $l$, it is easy to see that its length becomes $l + p(s) - p(t)$. Therefore the order between any two $s$-$t$ paths is preserved. 
    
    Observe that if we add an additional node $q$ to the graph with an edge of weight $0$ to every node in the original graph, then a valid potential function can be obtained by setting $p(u)$ to be the shortest distance from $q$ to $u$. We can run the Bellman-Ford algorithm from $q$ in $O(nm) = O(n ^ 3)$ to obtain this potential function.
\end{proof}
We can transform the graph during the preprocessing algorithm in the top layer. This does not affect the overall preprocessing time of $\tilde O(n ^ 3)$ for the top layer. From now on, we will assume that our graph does not contain negative edge weights at any time.

\section{The Basic Algorithm}
This section will be an introduction of the framework used in Gutenberg and Wulff-Nilsen's algorithm using slightly different approaches. Most importantly, we decided to merge the update procedure (the $\textsc{Delete}$ step in \cite{previous}) from individual layers into a single update procedure that integrates the results from all the layers for convenience in the later parts of our algorithm. We make sure that the layers are such that:
\begin{itemize}
    \item Each layer $i$ is preprocessed exactly per $2 ^ i$ updates.
    \item All layers are preprocessed initially. Therefore, every time a layer is preprocessed, all layers below it are preprocessed. 
    \item When multiple layers are preprocessed, their corresponding preprocessing steps are invoked from top to bottom (i.e. in decreasing order of indices).
    \item When layer $i > 0$ is preprocessed, it is given a set $C_i$ as a parameter, and it generates a set $C_{i - 1}$ used as the corresponding parameter for layer $i - 1$ to be processed right after. 
\end{itemize}
The implementation of layer $0$ will be specified later. Unless otherwise defined, let $G$ be the graph when the sub-routine in context is invoked, and for $i \in [0, L]$, we let $G_i$ be the snapshot of the graph taken when the layer $i$ was last preprocessed, and we let the set $D_i$ be the set of vertices deleted since layer $i$ was last preprocessed (and thus $G = G_i \backslash D_i$). When a graph is not explicitly stated in the context, assume the graph is $G$.

\subsection{Basic Preprocessing Algorithm}
Recall that on layer $i$, we need to compute paths passing through a set of centers $C_i$ subject to batch deletion of up to $2 ^ i$ vertices. To do this, we pre-compute some shortest paths through $C_i$ in the preprocessing procedure so that to deal with batch deletions, we can efficiently recover affected paths by concatenating some surviving paths. To achieve this, Gutenberg and Wulff-Nilsen designed a new way to use congestion. The congestion incurred by a path on a vertex is proportional to the running time cost we later pay when this path needs to be recovered. Hence the total congestion on a vertex is the running time cost that we will pay if this vertex is deleted. We define parameters $h_j := (3 / 2) ^ j$ for $j \in [0, i_h]$, where $i_h := \ceil{\log_{3 / 2}{n ^ {0.5}}}$. We let $H := h_{i_h} = \Theta(n ^ {0.5})$. For each $j \in [0, i_h]$, we compute the $h_j$-hop shortest paths through $C_i$ and add up the congestion contributions. The idea is that if a vertex is too congested, we delegate the computation of paths passing through it to the layer below by adding it to $C_{i - 1}$. To achieve this, any vertex with congestion more than a threshold is labeled congested and removed from the graph from future path computations, but we do not replace the already computed paths that may contain this vertex. We introduce this procedure in \cref{alg:basicpreprocessing} with a much slower implementation than the one introduced in their original paper for later convenience. For the bottom-most layer $0$ that maintains paths through $C_0$, there is no preprocessing step, and batch deletion step can be done by a simple recomputation in $\tilde O(n ^ 2C_0)$ time: since we assume the edge weights are non-negative, we can simply compute the shortest paths from and to each center by running Dijkstra's algorithm in $G$ and $\overleftarrow{G}$ and combine the answers.
\begin{algorithm} [H]
    \caption{Basic Preprocessing Algorithm} \label{alg:basicpreprocessing}
    \begin{algorithmic}[1]
        \Procedure{\textsc{BasicPreprocessing}}{$i, C_i$}
            \State $O_i \gets \emptyset$ 
            \State Set threshold ${\tau}_i \gets n ^ {2.5} / 2 ^ i$ 
            \For {$j \in [0, i_h]$}
                \State Initialize $\cg_i[j, v] \gets 0$ for all $v \in \ver$
                \For {$k \in [1, \abs{\ver}]$}
                    \State $s \gets \ver[k]$
                    \State $\pi_i[s, *, j] \gets \textsc{BellmanFordThroughCenter}(s, G \backslash O_i, C_i, h_j)$ \label{line:calcrestricted}
                    \For {$t \in V, u \in \pi_i[s, t, j]$}
                        \State $\cg_i[j, u] \gets \cg_i[j, u] + n / h_j$
                    \EndFor
                    \State $O_i \gets O_i \cup \{v \in \ver \mid \cg_i[j, v] > {\tau}_i\}$
                \EndFor
            \EndFor
            \State $C_{i - 1} \gets O_i$
        \EndProcedure
    \end{algorithmic}
\end{algorithm}
Note that the set $O_i$ is the same as $C_{i - 1}$ for now, but they will differ later in our final algorithm. For now, the procedure introduced has a lackluster running time of $\tilde O(n ^ 3H)$. In \cite{previous} a running time of $\tilde O(\abs{C_{i}}n ^ 2H)$ was achieved using randomization. We now introduce the lemma for the correctness of this preprocessing procedure.
\begin{lemma} \label{lemma:invariants}
    \Cref{alg:basicpreprocessing} preserves the following invariants:
    \begin{enumerate}
        \item $\forall (j, v) \in [0, i_h] \times \ver: \cg_i[j, v] \le 2{\tau}_i$,
        \item $\sum_{j \in [0, i_h], v \in \ver}{\cg_i[j, v]} = O(n ^ 3\log n)$,
        \item $\abs{O_i} = O(n ^ 3 \log n / {\tau}_i) = \tilde O(n ^ {0.5} 2 ^ i)$,
        \item $\pi_i[s, t, j]$ has a hop of at most $h_j$ and is not longer than the $h_j$-hop-shortest path between $s$ and $t$ in $G \backslash O_i$.
    \end{enumerate}
\end{lemma}
\begin{proof}
    To see (1), observe that a single iteration of the while loop does not contribute more than $n ^ 2 < {\tau}_i$ to the congestion on a single vertex. Since a vertex is removed from the path when its congestion exceeds $\tau_i$, the final congestion does not exceed $\tau_i + \tau_i = 2\tau_i$. To see (2), note that a path contributes exactly $n$ to the total congestion and there are $O(n ^ 2 \log n)$ paths. Therefore the total congestion does not exceed $n \times O(n ^ 2 \log n) = O(n ^ 3 \log n)$. To see (3), note that a vertex is added to $O_i$ only if its congestion value exceeds ${\tau}_i$. Since from (2) the total congestion is $O(n ^ 3 \log n)$, the claim follows. (4) is due to the fact that the graph $G \backslash O_i$ is decremental (i.e. we can only exclude more vertices from $G \backslash O_i$ the further we are into the procedure). 
\end{proof}

We define the total congestion on a vertex $v$ on layer $i$ to be $\cg_i(v)$. By definition, $\cg_i(v) = \sum_{0 \le j \le i_h}{\cg_i[j, v]}$. Since $i_h = O(\log n)$, we have $\cg_i(v) = \tilde O(\tau_i)$. 

\subsection{Basic Update Algorithm}
\subsubsection{Algorithm Description}
Although in the original framework from \cite{previous}, there is a separate batch deletion procedure for each layer. Our data structure uses a global update procedure that handles batch deletions on all the layers by polling the information on each layer. For the update procedure, at the start of the procedure, for each tuple $(s, v, t, j) \in \ver ^ 3 \times [0, i_h]$ we sample a random 0-1 variable $X[u, v, w, j]$ which equals $1$ with probability $\frac{100\log_2{n}}{h_j} = \tilde O(1 / h_j)$. The variables are fixed throughout the update. We will show in \cref{sec:space} how to sample these random variables implicitly space-efficiently. To do the update, we first check the effects of batch deletion on each layer. For layer $i$, if the removal of $D_i$ invalidates a computed path, we mark the path as a path that ``needs recovery.'' We look at the remaining paths on each layer. We let $\pi ^ {\prime}[s, t, j]$ be the shortest $\pi_i[s, t, j]$ that remains among all layers $i \in [0, L]$. We use these paths to recover the missing paths by using the random variables above to determine a hitting set\footnote{The algorithm in \cite{previous} did not rely on randomness in this part. It employed a deterministic hitting set construction.}. The update procedure is described in \cref{alg:basicupdate}.
\begin{algorithm} [H]
    \caption{Basic Update Procedure} \label{alg:basicupdate}
    \begin{algorithmic}[1]
        \Function{\textsc{Recover}}{$\pi_o, s, t, j$} \Comment{Recover ${\pi_o}[s, t, j]$}
            \State $\eta \gets \{v \mid X[s, v, t, j] = 1\} \cup \{s, t\}$ \Comment{Construct a hitting set.} \label{line:hsbasic}
            \State $x ^ {\prime} \gets \textrm{argmin}_{x \in \eta}{w({\pi_o}[s, x, j - 1] \circ {\pi_o}[x, t, j - 1])}$
            \State \Return ${\pi_o}[s, x ^ {\prime}, j - 1] \circ {\pi_o}[x ^ {\prime}, t, j - 1]$ \label{line:recover}
        \EndFunction
        \Function{\textsc{BasicUpdate}}{} \Comment{Compute the distance matrix after the update}
            \State $\pi ^ {\prime} \gets \perp$
            \State $\mathrm{NeedsRecovery}[*, *, *] \gets \mathtt{FALSE}$ \label{line:taggingrecovery}
            \For {$i \in [0, L]$} \label{line:forloopi}
                \For {$j \in [0, i_h]$}
                    \For {$s, t \in \ver$}
                        \If {$\pi_i[s, t, j] \cap D_i = \emptyset$} \label{line:checkintersect}
                            \If {$w(\pi_i[s, t, j]) < w({\pi ^ {\prime}}[s, t, j])$}
                                \State ${\pi ^ {\prime}}[s, t, j] \gets \pi_i[s, t, j]$ \label{line:getpiprime}
                            \EndIf
                        \Else
                            \State $\mathrm{NeedsRecovery}[s, t, j] = \mathtt{TRUE}$
                        \EndIf
                    \EndFor
                \EndFor
            \EndFor
            \State $\pi_o \gets \textrm{a copy of }\pi ^ {\prime}$
            \For {$j \in [0, i_h]$}
                \For {$s, t \in \ver$}
                    \If {$\mathrm{NeedsRecovery}[s, t, j] = \mathtt{TRUE}$}
                        \State ${\pi_o}[s, t, j] \gets \textsc{Recover}(\pi_o, s, t, j)$ \label{line:recoverbasic}
                    \EndIf
                \EndFor
            \EndFor
            \State $A \gets \textsc{RandGetShortestPaths}(G, H)$ \Comment{From \cref{lemma:long}} \label{line:geta}
            \For {$s, t \in \ver$}
                \State $A_{s, t} \gets \min(A_{s, t}, w({\pi_o}[s, t, i_h]))$ \label{line:combine}
            \EndFor
            \State \Return $A$
        \EndFunction
    \end{algorithmic}
\end{algorithm}

\subsubsection{Running Time Analysis}
For a triple $(s, t, j) \in \ver ^ 2 \times [0, i_h]$, define the \emph{recovery cost} $c(s, t, j) := n / h_j$. We first have the following lemma:
\begin{lemma} \label{lemma:coreprelude}
    The total time spent running line \ref{line:recoverbasic} is no more than the total recovery cost of triples $(s, t, j)$ such that $\mathrm{NeedsRecovery}[s, t, j] = \mathtt{TRUE}$ (up to a logarithmic factor) with high probability.
\end{lemma}
\begin{proof}
    It is easy to see that the $\textsc{Recover}(\pi_o, s, t, j)$ takes $\sum_{v}{X[s, v, t, j]} = O(\abs{\eta})$ time. Note that $X[u, v, w, j]$ equals $1$ with probability $\frac{100\log_2{n}}{h_j}$ and all random variables are independent. We have $O(\abs{\eta}) = \sum_{v}{X[s, v, t, j]} \le 10000\left(\log_2{n}\right) ^ 2c(s, t, j) = \tilde O(c(s, t, j))$ with high probability from the Chernoff bound\footnote{The multiplicative form of Chernoff bound is as follows: Let $x_1, x_2, \cdots, x_n$ be $n$ independent random variables taking values from $\{0, 1\}$ and let $x = \sum_i{x_i}$, and $\mu = E(x)$. Then, for any $\epsilon > 0$, we have $\mathbb{P}(x \ge (1 + \epsilon)\mu) \le e ^ {-\frac{\epsilon ^ 2\mu}{2 + \epsilon}})$ and $\mathbb{P}(x \le (1 - \epsilon)\mu) \le e ^ {-\frac{\epsilon ^ 2\mu}{2}})$.}.  Therefore, with high probability, the total time spent running line \ref{line:recoverbasic} is proportional to the total recovery cost of triples $(s, t, j)$ such that $\mathrm{NeedsRecovery}[s, t, j] = \mathtt{TRUE}$. 
\end{proof}

We now show the core argument used in Gutenberg and Wulff-Nilsen's framework with the following lemma:
\begin{lemma} \label{lemma:core}
    With high probability, the total time spent running line \ref{line:recoverbasic} is $\tilde O\left(\sum_{i}\sum_{v \in D_i}{\cg_i(v)}\right)$.
\end{lemma}
\begin{proof}
    The total recovery cost of triples $(s, t, j)$ such that $\mathrm{NeedsRecovery}[s, t, j] = \mathtt{TRUE}$ is no more than the total sum over $i \in [0, L]$ of the recovery cost of triples $(s, t, j)$ such that $\pi_i[s, t, j] \cap D_i \ne \emptyset$. From the preprocessing algorithm, we can see that $\cg_i(v)$ is equal to the total recovery cost of paths through $v$. Therefore the total recovery cost of triples $(s, t, j)$ such that $\mathrm{NeedsRecovery}[s, t, j] = \mathtt{TRUE}$ is no more than $\sum_{i}\sum_{v \in D_i}{\cg_i(v)}$. From \cref{lemma:coreprelude}, the total time spent running line \ref{line:recoverbasic} is $\tilde O\left(\sum_{i}\sum_{v \in D_i}{\cg_i(v)}\right)$ with high probability.
\end{proof}

\begin{lemma} \label{lemma:basicupdaterunningtime}
    \cref{alg:basicupdate} runs in $\tilde O(n ^ {2.5})$ time.
\end{lemma}
\begin{proof}
    From \cref{lemma:core}, the running time of the update procedure is 
    \begin{align*}
        \tilde O\left(\sum_{i}\sum_{v \in D_i}{\cg_i(v)}\right)     &= \tilde O\left(\sum_{i}\abs{D_i} \times 4\tau_i\right) \\
                                                                    &= \tilde O\left(\sum_{i}\abs{D_i}\tau_i\right) \\
                                                                    &= \tilde O\left (\sum_{i}2 ^ i \times n ^ {2.5} / 2 ^ i \right) \\
                                                                    &= \tilde O\left (\sum_{i}n ^ {2.5} \right) \\ 
                                                                    &= \tilde O(n ^ {2.5}).
    \end{align*}
    Since $H = n ^ {0.5}$, from \cref{lemma:long} the call to $\textsc{RandGetShortestPaths}$ runs in $\tilde O(n ^ {2.5})$ time. The line \ref{line:checkintersect} can be done by simply retrieving all the vertices on the paths and checking whether they are in $D_i$. Note that hops of paths do not exceed $O(H)$, and therefore this takes $\tilde O(n ^ 2H) = \tilde O(n ^ {2.5})$ time. The rest of the procedure can clearly be implemented in $\tilde O(n ^ {2.5})$ time. 
\end{proof}

\subsubsection{Correctness}
The next three lemmata show the correctness of \cref{alg:basicupdate}.
\begin{lemma} \label{lemma:inductionbase}
    If $\mathrm{NeedsRecovery}[s, t, j] = \mathtt{FALSE}$, then $\pi ^ {\prime}[s, t, j]$ is the $h_j$-hop shortest $s$-$t$ path in $G$. Moreover, if we constrain the for loop on \ref{line:forloopi} to $i \in [i ^ {\prime}, L]$ for some $i ^ {\prime} \in [0, L]$, then $\pi ^ {\prime}[s, t, j]$ is not longer than the $h_j$-hop shortest $s$-$t$ path through $C_{i ^ {\prime}}$ in $G$.
\end{lemma}
\begin{proof} 
    From invariant (4) in \cref{lemma:invariants}, for every $i \in [0, L]$, $\pi_i[s, t, j]$ is not longer than the $h_j$ shortest $s$-$t$ path through $C_i$ in $G_i \backslash O_i$. Since $G_i \backslash O_i \supset G_i \backslash C_{i - 1} \supset G \backslash C_{i - 1}$, $\pi_i[s, t, j]$ is not longer than the $h_j$ shortest $s$-$t$ path through $C_i$ in $G \backslash C_{i - 1}$. Since $\mathrm{NeedsRecovery}[s, t, j] = \mathtt{FALSE}$, every $\pi_i[s, t, j]$ is present. Therefore from a telescoping argument, $\pi ^ {\prime}[s, t, j]$ is not longer than the $h_j$-hop shortest $s$-$t$ path. Since from the preprocessing algorithm, $\pi ^ {\prime}[s, t, j]$ has a hop at most $h_j$, $\pi ^ {\prime}[s, t, j]$ is the $h_j$-hop shortest $s$-$t$ path.

    The proof of the second part of the lemma is similar.
\end{proof}

\begin{lemma} \label{lemma:inductionstep}
    Given $j \in [0, i_h]$, suppose for all pairs $(u, v) \in \ver ^ 2$, $\pi_o[u, v, j - 1]$ is not longer than the $h_{j - 1}$-hop-shortest $u$-$v$ path. Then for a fixed $(s, t) \in \ver ^ 2$, $\pi_o[s, t, j]$ is not longer than the $h_j$-hop shortest path between $s$ and $t$ with probability at least $1 - 1 / n ^ {10}$.
\end{lemma}
\begin{proof}
    Let $p$ be the $h_j$-hop shortest path between $s$ and $t$. We need to show that $w(\pi_o[s, t, j])\le w(p)$ with probability at least $1 - 1 / n ^ {10}$.
    
    If $\mathrm{NeedsRecovery}[s, t, j] = \mathtt{FALSE}$, then $w(\pi_o[s, t, j]) = w(\pi ^ {\prime}[s, t, j]) \le w(p)$ from \cref{lemma:inductionbase}. 
    
    Suppose $\mathrm{NeedsRecovery}[s, t, j] = \mathtt{TRUE}$. Let the $h_j$-hop-shortest $s$-$t$ path be $p$. If $p$ has a hop more than $h_{j - 1}$, then for any vertex $u$ on $p[\abs{p} - h_{j - 1}, h_{j - 1}]$, both $\abs{p[<u]}$ and $\abs{p[>u]}$ are no more than $h_{j - 1}$. Therefore by the assumption in the lemma, $w({\pi_o}[s, u, i - 1]) \le w(p[<u])$ and $w({\pi_o}[u, t, i - 1]) \le w(p[>u])$. which means that there are at least $h_{j - 1} - (h_j - h_{j - 1}) \ge 1 / 3h_j$ candidates $u$ such that $w({\pi_o}[s, u, i - 1] \circ {\pi_o}[u, t, i - 1]) \le w(p)$. From the Chernoff bound, with probability at least $1 - 1 / n ^ {10}$, at least one such candidate $u$ is such that $X[s, u, t, j] = 1$, which will enforce that $w(\pi_o[s, t, j]) \le w(p)$. For the case when $p$ has a hop of at most $h_{j - 1}$, since both $s$ and $t$ are included in the set $\eta$ on line \ref{line:hsbasic}, and from the assumption in the lemma $w(\pi_o[s, t, j - 1]) \le w(p)$, setting $x$ to be either $s$ or $t$ on line \ref{line:recover} enforces that $w(\pi_o[s, t, j]) \le w(p)$
\end{proof}

\begin{lemma} \label{lemma:induction}
    At the end of \cref{alg:basicupdate}, with a probability of at least $1 - 1 / n ^ {7.99}$, for all $(s, t, j) \in \ver ^ 2 \times [0, i_h]$, $\pi_o[s, t, j]$ is the $h_j$-hop-shortest path.
\end{lemma}
\begin{proof}
    We do induction on $j$ in increasing order and apply \cref{lemma:inductionstep}. By union bound over all $\tilde O(n ^ 2)$ recoveries throughout the update, our procedure is correct with probability at least $1 - 1 / n ^ {7.99}$.
\end{proof}

\begin{lemma} \label{lemma:basiccorrectness}
    At the end of \cref{alg:basicupdate}, with high probability, for every $(s, t) \in \ver ^ 2$, $A_{s, t}$ is the length of the shortest $s$-$t$ path in $G$. 
\end{lemma}
\begin{proof}
     By \cref{lemma:induction}, with high probability, for every $(s, t) \in \ver ^ 2$, $\pi_o[s, t, i_h]$ is not longer than the $H$-hop shortest $s$-$t$ path. It is easy to see from our procedure that $\pi_o[s, t, i_h]$ is a valid $s$-$t$ path, and therefore is not shorter than the shortest $s$-$t$ path. From \cref{lemma:long}, on line \ref{line:geta}, $A_{s, t}$ is not shorter than the shortest $s$-$t$ path, and if the shortest $s$-$t$ path has a hop more than $H$, $A_{s, t}$ is the length of the shortest $s$-$t$ path with high probability. Due to line \ref{line:combine}, at the end of the algorithm, with high probability, for every $(s, t) \in \ver ^ 2$, $A_{s, t}$ is the length of the shortest $s$-$t$ path.
\end{proof}

\section{An Efficient Data Structure With An Oracle} \label{sec:step1}
We will now introduce a data structure that achieves our desired running time with the help of an oracle that we will later dispose of. An $h$-hop shortest $s$-$t$ path is called \emph{$h$-hop-dominant} if it's also $2h$-hop shortest. A path $p$ is \emph{strongly-hop-dominant} if it's $4\abs{p}$-hop shortest. We will later prove the following lemma in \cref{sec:ssahdp} which states that strongly-hop-dominant paths can be computed efficiently:
\begin{restatable}{lemma}{computedominant} \label{lemma:computedominant}
   For a graph $G$ of $n$ vertices and $m$ edges and a starting vertex $s$ and a hop limit $H$, there is an algorithm $\textsc{SSAHDP}(G, H, s)$ that can, in $\tilde O(n + m)$ time and space, compute a set of paths from $s$ in $G$ that contains all strongly-hop-dominant paths starting from vertex $s$ with a hop at most $H$.
\end{restatable}

\subsection{Why Can We Not Just Replace Hop-Restricted Paths with Hop-Dominant Paths?} \label{sec:fail}
Suppose that we can somehow efficiently compute hop-dominant paths through a set of centers. It seems like we could simply replace line \ref{line:calcrestricted} from \cref{alg:basicpreprocessing} with hop-dominant paths. However, this turns out to be problematic. Note that the efficiency of batch deletion procedure is contingent on the following fact:
\begin{fact} \label{fact:fail}
    If the final shortest $s$-$t$ path $p$ has a hop in $(h_{j - 1}, h_j]$, either $\Pi[s, t] = p$, or for some $i$ and some $j ^ {\prime}$ that is ``relatively close'' to $j$ (in the basic data structure we have $j ^ {\prime} = j$), $\pi_i[s, t, j ^ {\prime}] \intersect D_i \ne \emptyset$.
\end{fact}

For hop-dominant paths, we will illustrate a simple counter-example to this. For simplicity, suppose that we only have the top layer $L$ (it is possible, but somewhat tedious, to extend this to multi-layers). Suppose that the shortest $s$-$t$ path $p_{10000}$ has a hop of $10000$. Suppose that for every hop $h \in [201, 9999]$, the shortest $h$-hop-restricted $s$-$t$ path $p_h$ is strictly longer than the shortest $(h + 1)$-hop-restricted $s$-$t$ path $p_{h + 1}$. Then there is no hop-dominant path with a hop in $(100, 10000)$. Suppose that the shortest $100$-hop-restricted $s$-$t$ path $p_{100}$ is a hop-dominant path. Then for every $j$ such that $100 \le h_{j - 1} < h_j < 10000$, the shortest $h_j$-hop-dominant path is $p_{100}$. However, we can have a scenario where:
\begin{itemize}
    \item $p_{100}$ survives the batch deletion,
    \item $p_{1000}$ survives the batch deletion, and
    \item For every $1000 < h \le 10000$, $p_h$ does not survive the batch deletion.
\end{itemize}
Now we can see that the final shortest $s$-$t$ path is $p_{1000}$, but $\Pi[s, t] = p_{10000}$. For $j$ such that $h_{j - 1} < 1000 \le h_j$, we have $\pi_L[s, t, j] = p_{100}$, which survives the batch deletion. Therefore, \cref{fact:fail} does not apply to hop-dominant paths.

In our algorithm, we will bypass this obstacle in two steps: In this section, we first efficiently ``approximate'' the congested vertex sets $C_i$ using hop-restricted paths. Then we show that if we are given the hop and length of each $s$-$t$ path, we can actually ``witness'' them by concatenating hop-dominant paths. Such proof is made possible by comparing our data structure with a hypothetical data structure that obtains all the hop-restricted paths. In \cref{sec:step2} we show how to get rid of the oracle. 

Finally, we leave it as an open problem whether a more straightforward approach is possible.
\subsection{Preprocessing Algorithm}
To speed up the preprocessing steps, we will employ two ideas. The first idea is to use random sampling to compute the set $O_i$. Instead of going through all the vertices and computing their contribution to congestion as we did in \cref{alg:basicpreprocessing}, we only do so for some vertices sampled randomly, and we try to simulate the original algorithm in an approximate sense by magnifying the congestion contributions. The second idea is that instead of computing all hop-restricted paths, we use \cref{lemma:computedominant} to compute only those that are strongly hop-dominant. With the oracle, we will show that these paths are sufficient for an efficient update procedure.

\subsubsection{Random Sampling} \label{sec:randomization}
The sampling step is shown in \cref{alg:sample}. The difference between \cref{alg:sample} and \cref{alg:basicpreprocessing} is that in \cref{alg:sample}, we skip the majority of vertices but we magnify the congestion contributions from paths to vertices by a large factor which is the inverse of the probability of keeping a vertex.
\begin{algorithm} [H]
    \caption{Sampling congested vertex set $O_i$ for paths through centers $C_i$}   \label{alg:sample}
    \begin{algorithmic}[1]
        \Function{\textsc{Sample}}{$i, C_i$}
            \State $O_i \gets \emptyset$
            \State Set threshold ${\tau}_i \gets n ^ {2.5} / 2 ^ i$ 
            \For {$j \in [0, i_h]$}
                \State Initialize $\cg_i[j, v] \gets 0$ for all $v \in \ver$
                \State $p_j \gets \frac{n ^ 2 \ln {n}}{\tau_i h_j}$ \Comment{Probability of keeping a vertex}
                \State Uniformly sample indices in $[1, \abs{\ver}]$ with probability $p_j$. \label{line:sampleindices}
                \For {$k \in [1, \abs{\ver}]$} \label{line:forloopk}
                    \State $s \gets \ver[k]$
                    \If {Index $k$ is sampled}
                        \State $P[*] \gets \textsc{BellmanFordThroughCenter}(s, G \backslash O_i, C_i, h_j)$ \label{line:getpaths}
                        \For {$t \in V, u \in P[t]$} \label{line:forloopcg}
                            \State $\cg_i[j, u] \gets \cg_i[j, u] + (n / h_j) / p_j$ \label{line:increasecg}
                        \EndFor
                        \State $O_i \gets O_i \cup \{v \in \ver \mid \cg_i[j, v] > {\tau}_i\}$ \label{line:addtoc}
                    \EndIf
                \EndFor
            \EndFor
            \State \Return $O_i$
        \EndFunction
    \end{algorithmic}
\end{algorithm}

We now show that \cref{alg:sample} indeed approximates a good set of congested vertices and runs efficiently. We analyze the running time, bound the size of $C_i$, and argue that our set of congestion vertices bounds the congestion values in an approximate sense for a hidden set of hop-restricted paths. 
\begin{lemma}
    \cref{alg:sample} runs in $\tilde O(\frac{n ^ 5}{\tau_i}) = \tilde O(n ^ {2.5}2 ^ i)$ time with high probability.
\end{lemma}
\begin{proof}
    For each $j$, observe that the number of indices sampled on line \ref{line:sampleindices} is $I_j = O(p_jn)$ with high probability due to the Chernoff Bound. Thus the running time of \cref{alg:sample} is $O(\sum_{j}{I_j \times n ^ 2h_j}) = O(\sum_{j}{p_jn ^ 3h_j}) = \tilde O(\frac{n ^ 5}{\tau_i})$ with high probability.
\end{proof}

Our congested vertex set satisfies the following:
\begin{lemma} \label{lemma:csize}
    At the end of \cref{alg:sample}, we have $\abs{O_i} = O(\frac{n ^ 3}{\tau_i})$ with high probability.
\end{lemma}
\begin{proof}
    For each level $j \in [0, i_h]$, since the number of indices sampled on line \ref{line:sampleindices} is $I_j = O(p_jn)$ and each source vertex contributes a congestion of $(n / h_j) / p_j$ to at most $n$ vertices, $\sum_{u}{\cg_i[j, u]} \le I_j \times (n ^ 2 / h_j) / p_j = O(n ^ 3 / h_j)$ with high probability. Since all congested vertices have a congestion value of at least $\tau_i$, the contribution to $\abs{C_i}$ on level $j$ is $O(\frac{n ^ 3}{\tau_ih_j})$ with high probability. Since $\sum_{j \in [0, i_h]}1 / h_j = \sum_{j \in [0, i_h]}{1 / (3 / 2) ^ j} = O(1)$, the contribution across all $j$ is $\sum_{j \in [0, i_h]}{\frac{n ^ 3}{\tau_ih_j}} = O(\frac{n ^ 3}{\tau_i})$ with high probability.
\end{proof}

We argue that the sampled vertices bound the congestion incurred by a set of paths that we could not afford to obtain. We consider a hypothetical data structure run in parallel with the actual data structure with the same random seeds. In the hypothetical data structure, suppose that we execute line \ref{line:getpaths} for every index $k$ instead of only for those that are sampled. We let $\pi_i[s, *, j] = P[*]$ computed on line \ref{line:getpaths}. This modified routine is very similar to \cref{alg:basicpreprocessing} except that we calculate congestion contribution in a randomized way. It is easy to see that in the hypothetical data structure, \cref{lemma:inductionbase} and Invariant (4) from \cref{lemma:invariants} are both still true. Consider the congestion $\overline{\cg}_i[*, *]$ computed in the normal way in the hypothetical data structure: for each index $k$ from the for loop on line \ref{line:forloopk}, we run the for loop on line \ref{line:forloopcg} and add $n / h_j$ to $\overline{\cg}_i[j, u]$. We will argue that the accumulation on $\overline{\cg}_i[j, u]$ is bounded with high probability for every $(j, u) \in [0, i_h] \times \ver$.
\begin{lemma} \label{lemma:implicit}
    In the hypothetical data structure, for any $c > 0$, with probability $1 - n ^ {-c}$, for every $(j, u) \in [0, i_h] \times \ver$, $\overline{\cg}_i[j, u] \le (2c + 9)\tau_i$.
\end{lemma}
\begin{proof}
    We first show that $\overline{\cg}_i[j, u] \le (2c + 9)\tau_i$ is unconditionally true as long as $\overline{\cg}_i[j, u] \le (2c + 8)\tau_i$ when $\cg_i[j, u] \le \tau_i$. To see this, suppose $\cg_i[j, u] > \tau_i$. Consider the last execution of the for loop on line \ref{line:forloopk} (with the modification mentioned above) before $\cg_i[j, u] > \tau_i$. Before this execution $\cg_i[j, u] \le \tau_i$. If $\overline{\cg}_i[j, u] \le (2c + 8)\tau_i$, since a single source contributes at most $n ^ 2 / h_j < \tau_i$ to $\overline{\cg_i[j, u]}$ and $u$ is removed from the graph after the execution, $\overline{\cg}_i[j, u] \le (2c + 9)\tau_i)$ will remain true.
    
    Therefore, it suffices to prove that with probability $1 - n ^ {-c}$, for every $u$, whenever $\cg_i[j, u] \le \tau_i$, $\overline{\cg}_i[j, u] \le (2c + 8)\tau_i$. We prove that with probability $1 - n ^ {-c}$, the contra-positive is true: for every $u$, whenever $\overline{\cg}_i[j, u] > (2c + 8)\tau_i$, $\cg_i[j, u] > \tau_i$.
    
    Assume $\overline{\cg}_i[j, u] > (2c + 8)\tau_i$. Let $c_0 = n ^ 2 / h_j$. We let $Y_k$ be the value between $[0, 1]$ that equals $p_j / c_0$ times the contribution to $\overline{\cg}_i[j, u]$ from $s = \ver[k]$, and let $X_k$ be the random variable between $[0, 1]$ that is equal to $Y_k / p_j$ if $k$ is sampled and $0$ otherwise. We can see that $\cg_i[j, u] = c_0 / p_j\sum_k{X_k}$ and $\overline{\cg}_i[j, u] = c_0 / p_j\sum_k{Y_k}$. $\cg_i[j, u] > \tau_i$ is equivalent to $c_0 / p_j\sum_k{X_k} > \tau_i$, which is then equivalent to $\sum_k{X_k} > \ln n$. Under the assumption that $$\overline{\cg}_i[j, u] = c_0 / p_j\sum_k{Y_k} > (2c + 8)\tau_i,$$ we have $$\sum_k{Y_k} > (2c + 8)\tau_i \times p_j / c_0 = (2c + 8)\ln{n}.$$ Since the indices are sampled independently, we have $$\mathbb{E}(X_k \mid X_1 \cdots X_{k - 1}, Y_1, \cdots Y_{k}) = p_j \times (Y_k / p_j) + (1 - p_j) \times 0 = Y_k$$ . We can now apply \cref{lemma:onlinechernoff} on $X$ and $Y$ with $\eps = 1 / 2$ and $\delta = (c + 3)\ln n$, which proves that $\sum_k{X_k} > \ln n$ does not hold with probability at most $n ^ {-c - 3}$. By taking the union bound over triples $(j, u, k) \in [0, i_h] \times \ver \times [1, \abs{\ver}]$, which there are $n ^ 2i_h \le n ^ 3$ of, we prove the claim.
\end{proof}

\subsubsection{New Preprocessing Algorithm}
Our new preprocessing algorithm combines the idea of computing only strongly dominant paths and the idea of applying random sampling. We first obtain the set $O_i$ using the sampling sub-routine above. Then for every $(s, t, j) \in \ver ^ 2 \times [0, i_h]$, we aim to compute a path $\Pi_i[s, t, j]$ with a hop of at most $h_j$ that is not longer than the $h_j$-hop shortest strongly-hop-dominant path between $s$ and $t$ through $C_i$ in $G \backslash O_i$. To do this, we compute all the strongly-hop-dominant paths to and from a vertex in $C_i$ in $G \backslash O_i$ using \cref{lemma:computedominant}, and combine them to cover all strongly-hop-dominant paths through $C_i$ in $G \backslash O_i$. For an positive interger $x$ we let $h ^ {-1}(x)$ be the smallest $y$ such that $x \le h_y$. For a path $p$, let the \emph{hop level} of $p$ be $h ^ {-1}(\abs{p})$. Our preprocessing algorithm is described in \cref{alg:newpreprocessing}.
\begin{algorithm} [H]
    \caption{New Preprocessing Algorithm} \label{alg:newpreprocessing}
    \begin{algorithmic}[1]
        \Procedure{\textsc{PreprocessingNew}}{$i, C_i$}
            \State $O_i \gets \textsc{Sample}(i, C_i)$
            \State $\Pi_i[*, *, *] \gets \perp$
            \For {$c \in C_i$}
                \State $P_{\textrm{from}} \gets \textsc{SSAHDP}(G \backslash O_i, H, c)$
                \State $P_{\textrm{to}} \gets \textsc{SSAHDP}(\overleftarrow{G \backslash O_i}, H, c)$
                \For {$p_{\textrm{from}} \in P_{\textrm{from}}, p_{\textrm{to}} \in P_{\textrm{to}}$}
                    \State $p \gets \overleftarrow{p_{\textrm{to}}} \circ p_{\textrm{from}}$
                    \State $j \gets h ^ {-1}(\abs{p})$
                    \If {$w(p) < w(\Pi_i[p[0], p[\abs{p}], j])$}
                        \State $\Pi_i[p[0], p[\abs{p}], j] \gets p$
                    \EndIf
                \EndFor
            \EndFor
            \State $C_{i - 1} \gets O_i$
        \EndProcedure
    \end{algorithmic}
\end{algorithm}

Due to \cref{lemma:computedominant}, \cref{alg:newpreprocessing} obviously runs in $\tilde O(n ^ 2\abs{C_i})$ time. To prove that the algorithm correctly computes $\Pi_i[*, *, *]$, we first prove the following lemma as a tool:
\begin{lemma} \label{lemma:combineshd}
    If a path $p$ is strongly-hop-dominant, for every $k \in [0, \abs{p}]$, both $p[0, k]$ and $p[k, \abs{p}]$ are strongly-hop-dominant.
\end{lemma}
\begin{proof}
    Suppose without loss of generality $p[0, k]$ is not strongly-hop-dominant. There is a path $p ^ {\prime}$ with a hop no more than $4k$ from $p[0]$ to $p[k]$ that is shorter than $p[0, k]$. By replacing $p[0, k]$ with $p ^ {\prime}$ we find a shorter path with the same starting and ending vertices as $p$ with a hop no more than $\abs{p} + 3k \le 4\abs{p}$, which contradicts the fact that $p$ is strongly-hop-dominant.
\end{proof}

Now we can show that the computed paths indeed have the desired property.
\begin{lemma} \label{lemma:correctpi}
    The path $\Pi_i[s, t, j]$ has a hop of at most $h_j$ and is not longer than the $h_j$-hop shortest strongly-hop-dominant path between $s$ and $t$ through $C_i$ in $G \backslash O_i$.
\end{lemma}
\begin{proof}  
    Straightforward from the procedure due to \cref{lemma:computedominant} and \cref{lemma:combineshd}.
\end{proof}

\subsection{Update Algorithm with An Oracle} \label{sec:oracle}
\subsubsection{Algorithm Description}
For each update, suppose we have an oracle $\textsc{HopAndLength}(s, t)$ that tells us the hop and the length of the shortest path between each pair of vertices $(s, t) \in \ver ^ 2$. We argue that we can use the oracle to concatenate paths computed in the preprocessing procedures to efficiently retrieve the shortest paths that have a hop at most $H$. To do this, for every $(s, t) \in \ver ^ 2$, we find ${\Pi ^ {\prime}}[s, t]$ which is the shortest surviving $s$-$t$ path computed from preprocessing each layer. Then for every $(s, t) \in \ver ^ 2$, if the shortest surviving $s$-$t$ path is not optimal, we recover it using the hitting set for the hop level of the shortest $s$-$t$ path, as detailed in \cref{alg:oracleupdate}.
\begin{algorithm} [H]
    \caption{Update Procedure With Oracle} \label{alg:oracleupdate}
    \begin{algorithmic}[1]
        \Function{\textsc{Recover}}{$\Pi_o, s, t, j$} \Comment{Recover ${\Pi_o}[s, t]$ using the hitting set for hop level $j$}
            \State $\eta \gets \{v \mid X[s, v, t, j] = 1\} \cup \{s, t\}$ \Comment{Construct a hitting set.}
            \State $x ^ {\prime} \gets \textrm{argmin}_{x \in \eta}{w({\Pi_o}[s, x] \circ {\Pi_o}[x, t])}$
            \State \Return ${\Pi_o}[s, x ^ {\prime}] \circ {\Pi_o}[x ^ {\prime}, t]$
        \EndFunction
        \Function{\textsc{UpdateWithOracle}}{}
            \State $\Pi ^ {\prime}[*, *] \gets \perp$ \label{line:oraclegetpiprimebegin}
            \For {$i \in [0, L]$}
                \For {$j \in [0, i_h]$}
                    \For {$s, t \in \ver$}
                        \If {$\Pi_i[s, t, j] \cap D_i = \emptyset$}
                            \If {$w(\Pi_i[s, t, j]) < w({\Pi ^ {\prime}}[s, t])$}
                                \State ${\Pi ^ {\prime}}[s, t] \gets \Pi_i[s, t, j]$
                            \EndIf
                        \EndIf
                    \EndFor
                \EndFor
            \EndFor \label{line:oraclegetpiprimeend}
            \State $\Pi_o = \textrm{a copy of }\Pi ^ {\prime}$
            \For {$j \in [0, i_h]$} \label{line:oracleforloop}
                \For {$s, t \in \ver$}
                    \State $(x, y) \gets \textsc{HopAndLength}(s, t)$
                    \If {$h ^ {-1}(x) = j\ \AND w({\Pi_o}[s, t]) > y$} \label{line:conditionoracle}
                        \State ${\Pi_o}[s, t] \gets \textsc{Recover}(\Pi_o, s, t, j)$ \label{line:oracleupdate}
                    \EndIf
                \EndFor
            \EndFor
            \State $A = \textsc{RandGetShortestPaths}(G, H)$
            \For {$s, t \in \ver$}
                \State $A_{s, t} \gets \min(A_{s, t}, w({\Pi_o}[s, t]))$
            \EndFor
            \State \Return $A$
        \EndFunction
    \end{algorithmic}
\end{algorithm}

\subsubsection{Running Time Analysis}
Recall the hypothetical data structure mentioned before where we obtain the full array of $\pi_i[*, *, *]$ for each layer $i$ in \cref{alg:sample} by executing line \ref{line:getpaths} for every index $k$. Assume that in that hypothetical data structure, for every update, instead of calling \cref{alg:oracleupdate} we call \cref{alg:basicupdate}. Recall that in the hypothetical data structure, \cref{lemma:inductionbase} and Invariant (4) from \cref{lemma:invariants} are both still true. We show that \cref{alg:oracleupdate} runs in $\tilde O(n ^ {2.5})$ time by relating its total recovery cost to the total recovery cost of \cref{alg:basicupdate} in the hypothetical data structure. We prove the following Lemma:
\begin{lemma} \label{lemma:updatecostbound}
    If the unique shortest $s$-$t$ path $p$ has a hop level $j \le i_h$ and gets recovered on line \ref{line:oracleupdate} of \cref{alg:oracleupdate}, then in the hypothetical data structure, $\mathrm{NeedsRecovery}[s, t, j ^ {\prime}]$ is TRUE in \cref{alg:basicupdate} for some $j ^ {\prime} \in \{j, j + 4\}$.
\end{lemma}
\begin{proof}
    Assume that $p$ gets recovered on line \ref{line:oracleupdate} of \cref{alg:oracleupdate} and that $\mathrm{NeedsRecovery}[s, t, j]$ is FALSE in \cref{alg:basicupdate}. By the second assumption and \cref{lemma:inductionbase}, in \cref{alg:basicupdate}, $\pi ^ {\prime}[s, t, j] = p$ and therefore we can find some $i \in [0, L]$ such that $\pi_i[s, t, j] = p$. By the first assumption, in \cref{alg:oracleupdate}, $w(\Pi ^ {\prime}[s, t]) > w(p)$ since otherwise the second condition on line \ref{line:conditionoracle} will never be true. Therefore $w(\Pi_i[s, t, j]) > w(p)$. From \cref{lemma:correctpi}, $p$ must not have been a strongly-hop-dominant $s$-$t$ path through $C_i$ in $G \backslash O_i$. From Invariant (4) in \cref{lemma:invariants}, $\pi_i[s, t, j + 4]$ is not longer than $h_{j + 4}$-hop shortest $s$-$t$ path through $C_i$ in $G \backslash O_i$. Since $h_{j + 4} > 4h_j$, $w(\pi_i[s, t, j + 4]) < w(p)$. Since $p$ is the shortest path after the batch deletion procedure, $\pi_i[s, t, j + 4]$ must have been deleted in the batch deletion in \cref{alg:basicupdate}, which makes $\mathrm{NeedsRecovery}[s, t, j + 4]$ TRUE.
\end{proof}

\begin{lemma} \label{lemma:oracleupdatecost}
    The total update cost for triples $(s, t, j)$ involved on line \ref{line:oracleupdate} is $\tilde O(n ^ {2.5})$ with high probability, and the running time of \cref{alg:oracleupdate} is $\tilde O(n ^ {2.5})$ with high probability.
\end{lemma}
\begin{proof}
    From \cref{lemma:updatecostbound} it is easy to see that the total recovery cost in \cref{alg:oracleupdate} is at most $(3 / 2) ^ 4$ times the recovery cost in \cref{alg:basicupdate} of the hypothetical data structure, which is $\sum_{i}\sum_{j}\sum_{u \in D_i}{\overline{\cg}_i[j, u]}$. From \cref{lemma:implicit} with high probability,
    \begin{align*}
        \sum_{i}\sum_{j}\sum_{u \in D_i}{\overline{\cg}_i[j, u]}     &\le \sum_{i}\sum_{j}\sum_{u \in D_i}{(2c + 9)\tau_i}  \\
                                                                    &= O\left(\sum_{i}i_h\abs{D_i}\tau_i\right) \\
                                                                    &= \tilde O\left(\sum_{i}\abs{D_i}\tau_i\right) \\
                                                                    &= \tilde O\left (\sum_{i}2 ^ i \times n ^ {2.5} / 2 ^ i \right) \\
                                                                    &= \tilde O\left (\sum_{i}n ^ {2.5} \right) \\ 
                                                                    &= \tilde O(n ^ {2.5}).
    \end{align*}
    which means that the total recovery cost in \cref{alg:oracleupdate} is also $\tilde O(n ^ {2.5})$ with high probability.

    The total running time of line \ref{line:oracleupdate} is proportional (up to a logarithmic factor) to the total recovery cost with high probability, which is $\tilde O(n ^ {2.5})$. It is easy to verify that the rest of \cref{alg:basicupdate} runs in $\tilde O(n ^ {2.5})$ time. Therefore, the running time of \cref{alg:basicupdate} is $\tilde O(n ^ {2.5})$ with high probability.
\end{proof}

\subsubsection{Correctness}
We now show the correctness of \cref{alg:oracleupdate}.
\begin{lemma} \label{lemma:oraclecorrect}
    With high probability, at the end of \cref{alg:oracleupdate}, for every $(s, t) \in \ver ^ 2$ such that the shortest $s$-$t$ path has a hop at most $H$, ${\Pi_o}[s, t]$ is the shortest $s$-$t$ path, and for every $(s, t) \in \ver ^ 2$, $A_{s, t}$ is the length of the shortest $s$-$t$ path.
\end{lemma}
\begin{proof}
    We do induction on the hop level $j$ of the shortest paths.
    
    Suppose for every $(s, t) \in \ver ^ 2$ such that the shortest $s$-$t$ path has a hop at most $h_{j - 1}$, ${\Pi_o}[s, t]$ is the shortest $s$-$t$ path. Then for the hop level $j$, if for some $(s, t) \in \ver ^ 2$ such that the shortest $s$-$t$ path has a hop between $h_{j - 1} + 1$ and $h_{j}$, then when the for loop on line \ref{line:oracleforloop} reaches hop level $j$, on line \ref{line:conditionoracle}, either the second condition is true and ${\Pi_o}[s, t]$ is recovered to be the $h_j$-hop shortest $s$-$t$ path due the same argument in the proof of \cref{lemma:inductionstep}, which by assumption is the shortest $s$-$t$ path, or the second condition is false and ${\Pi_o}[s, t] = \Pi ^ {\prime}[s, t]$ is already the shortest $s$-$t$ path. By induction, for every $(s, t) \in \ver ^ 2$ such that the shortest $s$-$t$ path has a hop at most $H$, ${\Pi_o}[s, t]$ is the shortest $s$-$t$ path.

    From a similar argument as the one in \cref{lemma:basiccorrectness}, with high probability, for every $(s, t) \in \ver ^ 2$, $A_{s, t}$ is the length of the shortest $s$-$t$ path.
\end{proof}

\subsection{Computing Single-Source Strongly-Hop-Dominant Paths} \label{sec:ssahdp}
Finally, we prove \cref{lemma:computedominant}.
\computedominant*
To introduce the algorithm $\textsc{SSAHDP}(G, H, s)$ that computes the single source strongly-hop-dominant paths, we first introduce a helper algorithm $\textsc{SSHDP}(G, h, s)$ that computes the $h$-hop-dominant paths for a given parameter $h$. Our algorithm is based on Dijkstra's single source shortest path algorithm. We find the unextended vertex with the shortest distance from the source and try to extend from that vertex, but different from Dijkstra's algorithm, we constrain that the hop does not exceed the given limit $h$ at any time. The details are shown in \cref{alg:sshdp}.
\begin{algorithm} [H]
    \caption{Single Source Hop Dominant Paths} \label{alg:sshdp}
    \begin{algorithmic}[1]
        \Function{\textsc{SSHDP}}{$G, h, s$}
            \For {$v \in \ver \backslash s$}
                \State {$\pi[v] \gets \perp$}
            \EndFor
            \State {$\pi[s] \gets \langle s\rangle$} 
            \State ${\ver}_0 \gets \ver$
            \While {${\ver}_0 \cap \{v \mid \pi[v] \ne \perp\} \ne \emptyset$}
                \State {$u \gets \argmin\{w(\pi[v ^ {\prime}]) \mid v ^ {\prime} \in {\ver}_0\}$}
                \State {${\ver}_0 \gets {\ver}_0 \backslash u$}
                \If {$\abs{\pi[u]} < h$}
                    \For {$v \in {\ver}_0$}
                        \If {$w(\pi[v]) > w(\pi[u] \circ \langle u, v \rangle)$ \OR ($w(\pi[v]) = w(\pi[u] \circ \langle u, v \rangle)$ \AND $\abs{\pi[v]} > \abs{\pi[u]} + 1$)} \label{line:sshdpcondition}
                            \State {$\pi[v] \gets \pi[u] \circ \langle u, v \rangle$}
                        \EndIf
                    \EndFor
                \EndIf
            \EndWhile
            \State \Return $\{\pi[v] \mid v \in \ver\} \backslash \perp$  
        \EndFunction
    \end{algorithmic}
\end{algorithm}
\cref{alg:sshdp} can be optimized using a Fibonacci Heap to $O(n \log n + m)$ time and $O(n + m)$ space on an $n$-vertex, $m$-edge graph, similar to the standard Dijkstra's Algorithm \cite{dijkstra, tarjanheap}. The output paths can be stored in $O(n)$ space if we store the paths as a rooted tree. Here since $m = O(n ^ 2)$, the running time is $\tilde O(n ^ 2)$, and the space usage is $O(n ^ 2)$.

\begin{proof} [Proof of \cref{lemma:computedominant}]
    We first show that in \cref{alg:sshdp} the set $\{\pi[v] \mid v \in V\}$ contains all the $h$-hop-dominant paths. We show that if an $h$-hop-dominant path exists from $s$ to $v$, then $\pi[v]$ is the $h$-hop-dominant path. Suppose an $h$-hop-dominant path $p$ exists between $s$ and $v$, but $\pi[v] \ne p$. Since we assume that hop-restricted shortest paths are unique and $\pi[v]$ is either $\perp$ or a path with a hop at most $h$, we have $w(\pi[v]) > w(p)$. Let $i > 0$ be the minimum index such that either $w(\pi[p[i]]) \ne w(p[0, i])$ or $\abs{\pi[p[i]]} > \abs{p[0, i]}$. Since $w(\pi[p[i - 1]]) = w(p[0, i - 1])$ and $\abs{\pi[p[i - 1]]} \le \abs{p[0, i - 1]} < h$, from the algorithm we know $w(\pi[0, p[i]]) \le w(p[0, i])$ and if $w(\pi[0, p[i]]) = w(p[0, i])$, $\abs{\pi[0, p[i]]} \le \abs{p[0, i]}$. From the choice of $i$, we have $w(\pi[0, p[i]]) < w(p[0, i])$. Therefore the path $\pi[p[i]] \circ p[i, h]$ has a shorter length than $p$ and has a hop of at most $h + h - i \le 2h$, which contradicts the fact that $p$ is $h$-hop-dominant.

    In order to include all single source strongly-hop-dominant paths with a hop \emph{at most} $H$, run \cref{alg:sshdp} with $h = 1, 2, 4, \cdots, 2 ^ {\ceil{\log_2{H}}}$. If a path $p$ is strongly-hop-dominant, then it is $2 ^ {\ceil{\log_2{\abs{p}}}}$-hop-dominant. Since there are $O(\log n) = \tilde O(1)$ possible values of $h$, the total running time is still $\tilde O(n ^ 2)$, and the output paths can be stored in $\tilde O(n)$ space. The space usage is $O(n ^ 2)$ since we can discard all temporary space usage after each run.
\end{proof}

\section{Final Data Structure Without The Oracle} \label{sec:step2}
Now we assume that we no longer have access to the oracle used in \cref{alg:oracleupdate} in the actual data structure. Suppose we run another hypothetical data structure (distinct from the one described earlier in \cref{sec:randomization}) in parallel with the actual data structure with the same random seeds but with access to the oracle. Throughout the section when we refer to \cref{alg:oracleupdate}, we refer to the one inside the hypothetical data structure. Consider \cref{alg:oracleupdate}. We can categorize the $(s, t) \in \ver ^ 2$ pairs where the shortest $s$-$t$ path has a hop at most $H$ into two types:
\begin{itemize}
    \item Primary: if $\Pi_o[s, t] = \Pi ^ {\prime}[s, t]$, that is, the shortest $s$-$t$ path has been present before the start of the for loop on line \ref{line:oracleforloop}, the pair $(s, t)$ and the path $\Pi_o[s, t]$ are called \emph{primary}.
    \item Secondary: if path $\Pi_o[s, t] \ne \Pi ^ {\prime}[s, t]$, that is, the shortest $s$-$t$ path is recovered on line \ref{line:oracleupdate}, the pair $(s, t)$ and the path $\Pi_o[s, t]$ are called \emph{secondary}.
\end{itemize}

In \cref{alg:oracleupdate}, every time we try to recover a secondary pair of vertices $(s, t)$ where the hop level of the shortest $s$-$t$ path is $j$, we go through candidates $v$ such that $X[s, v, t, j] = 1$ and use the concatenation of paths $\Pi_o[s, v]$ and $\Pi_o[v, t]$ as a candidate. There are two types of candidates here:
\begin{enumerate}
    \item Type I: both of $(s, v)$ and $(v, t)$ are primary. 
    \item Type II: one or both of $(s, v)$ and $(v, t)$ are secondary. 
\end{enumerate}
Note that ``primary'' and ``secondary'' are for $(s, t)$ pairs and ``type I'' and ``type II'' are for concatenations.

We will first introduce how we can deal with type I candidates, and then we will deal with type II candidates.

\subsection{Type I candidates} \label{sec:typei}
Given the paths $\Pi ^ {\prime}[*, *]$ computed in \cref{alg:oracleupdate}, for each pair of vertices $(s, t)$ we define an $s$-$t$ \emph{concatenation} to be any path of the form $\Pi ^ {\prime}[s, v] \circ \Pi ^ {\prime}[v, t]$ for some $v \in \ver$. Let $j$ be the hop level of the shortest path between $s$ and $t$. The concatenation is called \emph{hitting set compatible} if $X[s, v, t, j] = 1$. We consider the auxiliary problem of maintaining \emph{all-pairs shortest hitting set compatible concatenations} subject to vertex deletions. Specifically, for each pair of vertices $(s, t) \in \ver ^ 2$, we will find an $s$-$t$ path in the current graph that is at most as long as the shortest $s$-$t$ concatenation that is hitting set compatible.

This task is similar to the problem of maintaining APSP subject to vertex deletion, and it is very easy to design an auxiliary data structure similar to the previous data structure for dynamic APSP by Gutenberg and Wulff-Nilsen and attach it to our existing data structure. Intuitively, since we only concatenate two primary paths, it is almost as if we only need to maintain all-pairs 2-hop shortest paths. We now introduce how we implement such an auxiliary data structure in detail.

We will use the following lemma to be proven in \cref{sec:concat}:
\begin{restatable}{lemma}{concat} \label{lemma:concat}
    Given a set of paths $P$ with hop at most $h$ inside a graph $G := \langle \ver, \e \rangle$ with $n = \abs{\ver}$, and a vertex $c \in \ver$. Suppose that for every path $p \in P$, the following operations can be done in $O(1)$ time:
    \begin{itemize}
        \item Given a vertex $c \in \ver$, check whether $c \in p$;
        \item Given a vertex $c \in \ver$, find the path $p[<c]$;
        \item Given a vertex $c \in \ver$, find the path $p[>c]$.
    \end{itemize}
    Then there is an algorithm $\textsc{Concat}(P, c)$ in $\tilde O(\abs{P} + n ^ 2)$ time and space that can compute an array $\pi[*, *]$ where for any pair of vertices $(s, t) \in \ver ^ 2$, $\pi[s, t]$ is either an $s$-$t$ path with a hop of $O(h)$ in $G$ or $\perp$, such that for all $p_0, p_1 \in P$ such that $c \in p_0 \circ p_1$ and $p_0 \circ p_1$ is an $s$-$t$ path, $w(\pi[s, t]) \le w(p_0 \circ p_1)$. 
\end{restatable} 

\subsubsection{Final Preprocessing Algorithm}
For the preprocessing algorithms, our idea is to find a set of good concatenations such that after each batch deletion, only a small amount of these concatenations need to be recomputed. Whenever a layer $i$ is preprocessed, we first obtain $O_i$ and $\Pi_i[*, *, *]$ in the same way as in \cref{alg:newpreprocessing}. Then for each hop level $j$ and pair of vertices $(s, t) \in \ver ^ 2$, we perform an additional routine to cover the possible concatenations of some path $\Pi_{i_0}[s, v, j_0]$ ($(v, j_0) \in \ver \times [0, j]$) from layer $i_0 \in [i, L]$, and some path $\Pi_{i_1}[v, t, j_1]$ ($(t, j_1) \in \ver \times [0, j]$) from layer $i_1 \in [i, L]$, and we keep track of the congestion caused by these concatenations and remove \emph{additional congested vertices} in a similar fashion as the algorithms before. We adopt the randomization technique from Gutenberg and Wulff-Nilsen's paper. We go through the centers in $C_i$ in a uniformly random order. For every center, we use \cref{lemma:concat} to cover the concatenations involving that center avoiding the set of additional congested vertices stored in the set $\overline{O}_i$, and then we compute the congestion increase on the vertices and add the additional congested vertices to $\overline{O}_i$. Finally, we let $C_{i - 1}$ be the union of the set $O_i$ of congested vertices and the set $\overline{O}_i$ of additional congestedvertices. \textbf{This preprocessing algorithm is final as for type II candidates we do not need to take extra steps in the preprocessing algorithm.} The algorithm is described in \cref{alg:finalpreprocessing}.
\begin{algorithm} [H]
    \caption{Final Data Structure: Preprocessing c.f. Algorithm 4 in \cite{previous}} \label{alg:finalpreprocessing}
    \begin{algorithmic}[1]
        \Procedure{\textsc{Preprocessing}}{$i, C_i$}
            \State Compute $O_i$ and $\Pi_i[*, *, *]$ in the same way as in \cref{alg:newpreprocessing}
            \State $C ^ {\prime} \gets \emptyset$
            \State Set threshold ${\tau}_i \gets n ^ {2.5} / 2 ^ i$ 
            \State $\overline{\Pi}_i[*, *, *] \gets \perp$
            \State $\overline{O}_i \gets \emptyset$
            \State Initialize $\overline{\cg}_i[v] \gets 0$ for all $v \in \ver$
            \For {$c \in C_i$ in a uniformly random order}
                \Do
                    \For {$j \in [0, i_h]$}
                        \State $P \gets \{\Pi_{i ^ {\prime}}[s, t, j ^ {\prime}] \mid (i ^ {\prime}, s, t, j ^ {\prime}) \in [i, L] \times \ver ^ 2 \times [0, j] \wedge \Pi_{i ^ {\prime}}[s, t, j ^ {\prime}] \cap \overline{O}_i = \emptyset\}$
                        \State $\phi[*, *, j] \gets \textsc{Concat}(P, c)$
                    \EndFor 
                    \For {$(s, t, j) \in \ver ^ 2 \times [0, i_h]$}
                        \If {$w(\phi[s, t, j]) < \overline{\Pi}_i[s, t, j]$}
                            \State $\overline{\Pi}_i[s, t, j] \gets \phi[s, t, j]$
                            \For {$u \in \overline{\Pi}_i[s, t, j]$}
                                \State $\overline{\cg}_i[v] \gets \overline{\cg}_i[v] + n / h_j$ \label{line:addtooverlinecg}
                            \EndFor
                        \EndIf
                        \If {$\overline{O}_i \ne \{v \in \ver \mid \overline{\cg}[v] > {\tau}_i\}$}
                            \State $\overline{O}_i \gets \{v \in \ver \mid \overline{\cg}[v] > {\tau}_i\}$
                            \For {$j \in [0, i_h]$}
                                \State $P \gets \{\Pi_{i ^ {\prime}}[s, t, j ^ {\prime}] \mid (i ^ {\prime}, s, t, j ^ {\prime}) \in [i, L] \times \ver ^ 2 \times [0, j] \wedge \Pi_{i ^ {\prime}}[s, t, j ^ {\prime}] \cap \overline{O}_i = \emptyset\}$
                                \State $\phi[*, *, j] \gets \textsc{Concat}(P, c)$
                            \EndFor 
                        \EndIf
                    \EndFor
                \doWhile {$\overline{O}_i \ne \emptyset$}
            \EndFor
            \State $C_{i - 1} \gets O_i \cup \overline{O}_i$
        \EndProcedure
    \end{algorithmic}
\end{algorithm}

Firstly, we will show that the conditions in \cref{lemma:concat} are satisfied for $P$.
\begin{lemma} \label{lemma:satisfied}
    For every $(i, s, t, j) \in [0, L] \times \ver ^ 2 \times [0, i_h]$, let $p := \Pi_i[s, t, j]$. If $p \ne \perp$, the following operations can be done in $O(1)$ time without using more than $\tilde O(n ^ 2)$ space:
    \begin{itemize}
        \item Given a vertex $c \in \ver$, check whether $c \in p$;
        \item Given a vertex $c \in p$, find the path $p[<c]$;
        \item Given a vertex $c \in p$, find the path $p[>c]$.
    \end{itemize}
\end{lemma}
\begin{proof}
    Observe \cref{alg:finalpreprocessing} and note that $p$ is the concatenation of two paths that are outputs of \cref{alg:sshdp}. Since the outputs of \cref{alg:sshdp} can be stored as a rooted tree, it suffices to perform two queries on two rooted trees of the following type: given two nodes $c$ and $x$, check whether $c$ is on the path from the root to node $x$, and if so, find the position of $c$ on the path. This can be done using standard techniques on static $n$-node rooted trees in $\tilde O(1)$ time per query using $\tilde O(n)$ space\footnote{One such technique is to get the DFS traversal of the tree. The sub-tree of each node is an interval in the DFS traversal. The problem is equivalent to checking if $x$ is in the sub-tree of $c$, which can be done by simply checking if $x$ is in the interval of $c$ in the DFS traversal. The position of $c$ on the path is simply the depth of $c$ in the tree.}. Since there are $\tilde O(n)$ function calls to \cref{alg:sshdp}, there are $\tilde O(n)$ such trees in total, and the total space usage is $\tilde O(n ^ 2)$. To store $p[<c]$ and $p[>c]$, we can see that either of them is a concatenation of at most two paths on two trees. It suffices to simply store the pointers to the starting and the ending nodes on the trees.
\end{proof}

The reader can compare the following lemmata to Lemma 4.3 in \cite{previous}\footnote{Lemma 4.5 in the full version on arXiv and Lemma 4.3 in the proceedings version.}.
\begin{lemma} \label{lemma:wahrvawlfasfn}
    For some globally fixed constant $\gamma > 0$, with a probability at least $1 / 2$, \cref{alg:finalpreprocessing} preserves the following invariant: $$\sum_{v \in \ver}{\overline{\cg}_i[v]} \le \gamma n ^ 3(\log_2{n}) ^ 2.$$ 
\end{lemma}
\begin{proof}
    Firstly, for any $(s, t, j) \in \ver ^ 2 \times [0, i_h]$, for a fixed center $c$, $\overline{\Pi}_i[s, t, j]$ gets updated at most once since adding more vertices to $\overline{O}_i$ does not improve $\phi[s, t, j]$. From \cref{lemma:concat}, $\phi[s, t, j]$ has a hop of at most $O(h_j)$, and therefore every time $\overline{\Pi}_i[s, t, j]$ is updated, at most $O(h_j)$ vertices receive a congestion increase of $n / h_j$. Therefore for $(s, t, j)$, the total increase in total congestion due to center $c$ is at most $O(h_j) \times n / h_j = O(n)$.
    
    Fix $(s, t, j) \in \ver ^ 2 \times [0, i_h]$. Suppose the current center $c$ is the $k$-th chosen center. The probability that the $k$-th center updates $\overline{\Pi}_i[s, t, j]$ is the probability that the following event $E_0$ happens: $\phi[s, t, j]$ computed at this center $c$, using some congested set $\overline{O}_i$, is shorter than $\phi[s, t, j]$ computed at all previous centers, using what the set $\overline{O}_i$ was at those centers. Consider the following event $E_1$: $\phi[s, t, j]$ computed at this center $c$ using $\overline{O}_i$, is shorter than $\phi[s, t, j]$ computed at previous centers using $\overline{O}_i$ \emph{at the current center}. Since we only add vertices to $\overline{O}_i$, $\phi[s, t, j]$ computed from previous centers can only be the same or longer in $E_1$ compared to in $E_0$. Therefore event $E_0$ implies event $E_1$, and therefore $\mathrm{P}(E_0) \le \mathrm{P}(E_1)$. In $E_1$, since the set $\overline{O}_i$ is fixed, the sequence of computation does not matter and $\mathrm{P}(E_1)$ is equal to the probability that among the first $k$ centers, $\phi[s, t, j]$ is the shortest when computed at the $k$-th center using $\overline{O}_i$, and such probability is obviously $1 / k$. Thus $\mathrm{P}(E_0) \le \mathrm{P}(E_1) = 1 / k$. Let $Y[s, t, j]$ be the expected amount of congestion added on line \ref{line:addtooverlinecg}. We have $$E[Y[s, t, j]] \le \sum_{1 \le k \le \abs{C_i}}{1 / k \times O(n)} = O(nH_n) = O(n \log n),$$ where $H_n = O(\log n)$ is the $n$-th harmonic number. Since there are $O(\log n)$ choices of $j$'s and $O(n ^ 2)$ choices of $s$ and $t$, the expected total congestion $E[\sum_{v \in \ver}{\overline{\cg}_i[v]}] \le \frac{\gamma}{2} n ^ 3 (\log_2{n}) ^ 2$ for some globally fixed constant $\gamma > 0$. From Markov's inequality, $E[\sum_{v \in \ver}{\overline{\cg}_i[v]}] \le \gamma n ^ 3 (\log_2{n}) ^ 2$ with a probability of at least $1 / 2$.
\end{proof}

\begin{lemma} \label{lemma:invariantstypei}
    With a probability of at least $1 / 2$, \cref{alg:finalpreprocessing} preserves the following invariants:
    \begin{enumerate}
        \item $\forall (j, v) \in [0, i_h] \times \ver: \overline{\cg}_i[v] \le 2{\tau}_i$,
        \item $\sum_{v \in \ver}{\overline{\cg}_i[v]} = O(n ^ 3\log ^ 2 n)$, 
        \item $\abs{\overline{O}_i} = O(n ^ 3 \log ^ 2 n / {\tau}_i) = \tilde O(n ^ {0.5} 2 ^ i)$,
        \item For every $(s, t, j) \in \ver ^ 2 \times [0, i_h]$, $\overline{\Pi}_i[s, t, j])$ is an $s$-$t$ path with a hop of $O(h_j)$ that at most as long as the shortest $s$-$t$ path not visiting a vertex in $\overline{O}_i$ that is of the form $\Pi_{i_0}[s, v, j_0] \circ \Pi_{i_1}[v, t, j_1]$ for some $(i_0, i_1, v, j_0, j_1) \in [i, L] ^ 2 \times \ver \times [0, j] ^ 2$ where $c \in \Pi_i[v, t, j_1]$.
    \end{enumerate}
    \cref{alg:finalpreprocessing} runs in $\tilde O(n ^ {2.5} 2 ^ i)$ time.
\end{lemma}
\begin{proof}
    Since a single path increases the congestion on a vertex by at most $n$, $\overline{\cg}[v] < {\tau}_i + n < 2{\tau}_i$ for every vertex $v$, which proves (1). (2) is straightforward from \cref{lemma:wahrvawlfasfn}. Since each vertex $v \in \overline{O}_i$ is such that $\overline{\cg}_i[v] \ge {\tau}_i$, $\abs{\overline{O}_i} \le O(n ^ 3 \log ^ 2 n / {\tau}_i) = \tilde O(n ^ {0.5} 2 ^ i)$, which proves (3). (4) is not hard to see from the construction of the set of paths $P$ in the algorithm, the correctness of \cref{lemma:concat}, and the fact that the graph $G \backslash \overline{O}_i$ is decremental.

    For every recomputation of $\phi[*, *, *]$, $\abs{P} = \tilde O(n ^ 2)$. It is easy to see that \cref{lemma:concat} applies for $P$ from \cref{lemma:satisfied}. From \cref{lemma:concat} the recomputation takes $\tilde O(n ^ 2)$ time. Since before every recomputation of $\phi[*, *, *]$ either some vertex is added to $\overline{O}_i$ or we move on to the next center in $C_i$, there are $O(\abs{C_i} + \abs{\overline{O}_i})$ recomputations in total and the total running time for these recomputations is $\tilde O((\abs{C_i} + \abs{\overline{O}_i})n ^ 2) = \tilde O(n ^ {2.5} 2 ^ i)$. The rest of the algorithm clearly runs in $\tilde O(n ^ {2.5} 2 ^ i)$ time, which proves the claim.
\end{proof}
For any arbitrarily large constant $c ^ {\prime}$, we can repeat the procedure $c ^ {\prime} \log_2 {n} = \tilde O(1)$ times until the invariants are met, which magnifies the probability of success to $1 - n ^ {-c ^ {\prime}}$.

\subsubsection{Update Algorithm For Type I Candidates}
The update procedure for all-pairs shortest concatenation is similar to \cref{alg:basicupdate}, but recovery is done using two paths in $\Pi ^ {\prime}$. The update procedure is described in \cref{alg:updateconcatenation}.
\begin{algorithm} [H]
    \caption{Update Procedure for Concatenations Only} \label{alg:updateconcatenation}
    \begin{algorithmic}[1]
        \Function{\textsc{Recover}}{$\Pi ^ {\prime}, s, t, j$} \Comment{Recover ${\pi ^ {\prime}}[s, t]$ using the hitting set for hop level $j$}
            \State $\eta \gets \{v \mid X[s, v, t, j] = 1\} \cup \{s, t\}$ \Comment{Construct a hitting set.}
            \State $x ^ {\prime} \gets \textrm{argmin}_{x \in \eta}{w({\Pi ^ {\prime}}[s, x] \circ {\Pi ^ {\prime}}[x, t])}$ \label{line:recoveryhittingset}
            \State \Return ${\Pi ^ {\prime}}[s, x ^ {\prime}] \circ {\Pi ^ {\prime}}[x ^ {\prime}, t]$
        \EndFunction
        \Function{\textsc{UpdateForConcatenations}}{}
            \State Compute $\Pi ^ {\prime}$ using the part from line \ref{line:oraclegetpiprimebegin} to line \ref{line:oraclegetpiprimeend} from \ref{alg:oracleupdate}.
            \State $\overline{\Pi ^ {\prime}} \gets \perp$
            \State $\mathrm{NeedsRecoveryTypeI}[*, *, *] \gets \mathtt{FALSE}$ 
            \For {$i \in [0, L]$}
                \For {$j \in [0, i_h]$}
                    \For {$s, t \in \ver$}
                        \If {$\overline{\Pi}_i[s, t, j] \cap D_i = \emptyset$}
                            \If {$w(\overline{\Pi}_i[s, t, j]) < w({\overline{\Pi ^ {\prime}}}[s, t, j])$}
                                \State ${\overline{\Pi ^ {\prime}}}[s, t, j] \gets \overline{\Pi}_i[s, t, j]$ \label{line:fjakldgfvzskdlgjkselt}
                            \EndIf
                        \Else
                            \State $\mathrm{NeedsRecoveryTypeI}[s, t, j] = \mathtt{TRUE}$
                        \EndIf
                    \EndFor
                \EndFor
            \EndFor
            \For {$j \in [0, i_h]$}
                \For {$s, t \in \ver$}
                    \If {$\mathrm{NeedsRecoveryTypeI}[s, t, j] = \mathtt{TRUE}$}
                        \State ${\overline{\Pi ^ {\prime}}}[s, t, j] \gets \textsc{Recover}(\Pi ^ {\prime}, s, t, j)$ \label{line:recoveroverlinepi}
                    \EndIf
                \EndFor
            \EndFor
            \State \Return $\overline{\Pi ^ {\prime}}$
        \EndFunction
    \end{algorithmic}
\end{algorithm}
To execute line \ref{line:findargmin} in $\tilde O(1)$ time, one can use a Fibonacci Heap, similar to the standard Dijkstra's Algorithm \cite{dijkstra, tarjanheap}.

\begin{lemma}
    \Cref{alg:updateconcatenation} runs in $\tilde O(n ^ {2.5})$ time with high probability.
\end{lemma}
\begin{proof}
    This can be derived from a similar argument as the one used to prove \cref{lemma:basicupdaterunningtime}.
\end{proof}

The next lemma shows the correctness of \cref{alg:updateconcatenation}.
\begin{lemma} \label{lemma:typei}
    At the end of the procedure, for every $(s, t) \in \ver ^ 2$, if some shortest hitting set compatible $s$-$t$ concatenation is the shortest $s$-$t$ path, then $\overline{\Pi ^ {\prime}}[s, t, i_h]$ is the shortest $s$-$t$ path.
\end{lemma}
\begin{proof}
    Suppose some shortest hitting set compatible $s$-$t$ concatenation is the shortest $s$-$t$ path. Let $j$ be the hop level of the shortest $s$-$t$ path. Let some shortest hitting set compatible $s$-$t$ concatenation be $\Pi ^ {\prime}[s, v] \circ \Pi ^ {\prime}[v, t]$. Suppose $\Pi ^ {\prime}[s, v]$ comes from $\Pi_{i_0}[s, v, j_0]$ and $\Pi ^ {\prime}[v, t]$ comes from $\Pi_{i_1}[v, t, j_1]$. Since $\abs{\Pi ^ {\prime}[s, v]} \le h_j$, we have $j_0 \le j$. Similarly $j_1 \le j$. It suffices to show that $\overline{\Pi ^ {\prime}}[s, t, j]$ is the shortest $s$-$t$ path, since if this is the case, since on line \ref{line:recoveryhittingset} both $s$ and $t$ are included in $\eta$, $\overline{\Pi ^ {\prime}}[s, t, i_h]$ will also be the shortest $s$-$t$ path. 
    
    Suppose $\mathrm{NeedsRecoveryTypeI}[s, t, j] = \mathtt{FALSE}$. Let $i ^ {\prime}$ be the smallest index such that either $\Pi ^ {\prime}[s, v] \cap C_{i ^ {\prime}} \ne \emptyset$ or $\Pi ^ {\prime}[v, t] \cap C_{i ^ {\prime}} \ne \emptyset$. Then $\Pi ^ {\prime}[s, v] \cap C_{i ^ {\prime} - 1} = \emptyset$ and $\Pi ^ {\prime}[v, t] \cap C_{i ^ {\prime} - 1} = \emptyset$ which means that $\Pi ^ {\prime}[s, v] \cap O_{i ^ {\prime}} = \emptyset$ and that $\Pi ^ {\prime}[v, t] \cap O_{i ^ {\prime}} = \emptyset$ (For the case of $i = 0$, note that layer $0$ is done by brute force and therefore we can assume that $O_0 = \emptyset$). Note that $\Pi ^ {\prime}[s, v] = \Pi_{i_0}[s, v, j_0]$ is through $C_{i_0}$ and $\Pi ^ {\prime}[v, t] = \Pi_{i_1}[v, t, j_1]$ is through $C_{i_1}$, which means that $\Pi ^ {\prime}[s, v] \cap C_{i_0} \ne \emptyset$ and $\Pi ^ {\prime}[v, t] \cap C_{i_1} \ne \emptyset$. Therefore we have $i ^ {\prime} \le \min(i_0, i_1)$. From Invariant (4) in \cref{lemma:invariantstypei}, $\overline{\Pi}_{i ^ {\prime}}[s, t, j]$ is not longer than $\Pi ^ {\prime}[s, v] \circ \Pi ^ {\prime}[v, t]$. Since this path will be used to update $\overline{\Pi ^ {\prime}}[s, t, j]$ on line \ref{line:fjakldgfvzskdlgjkselt}, at the end of \cref{alg:updateconcatenation}, $\overline{\Pi ^ {\prime}}[s, t, j]$ is not longer than the shortest hitting set compatible $s$-$t$ concatenation. Since the shortest hitting set compatible $s$-$t$ concatenation is the shortest $s$-$t$ path and $\overline{\Pi ^ {\prime}}[s, t, j]$ is either $\perp$ or an $s$-$t$ path, $\overline{\Pi ^ {\prime}}[s, t, j]$ is the shortest $s$-$t$ path.

    If $\mathrm{NeedsRecoveryTypeI}[s, t, j] = \mathtt{TRUE}$, then line \ref{line:recoveroverlinepi} will be run for $(s, t, j)$. Since $X[s, v, t, j] = 1$, $v$ will be a candidate on line \ref{line:recoveryhittingset}, and therefore $w(\overline{\Pi ^ {\prime}}[s, t, j]) \le w(Pi ^ {\prime}[s, v] \circ Pi ^ {\prime}[v, t])$. At the end of \cref{alg:updateconcatenation}, since $Pi ^ {\prime}[s, v] \circ Pi ^ {\prime}[v, t]$ is the shortest $s$-$t$ path and $\overline{\Pi ^ {\prime}}[s, t, j]$ is either $\perp$ or an $s$-$t$ path, $\overline{\Pi ^ {\prime}}[s, t, j]$ is the shortest $s$-$t$ path.
\end{proof}
We will show in \cref{sec:typeii} that for $(s, t) \in \ver ^ 2$ where the hitting set compatible $s$-$t$ concatenation is not the shortest $s$-$t$ path, some type II candidate can be used to correctly recover the shortest $s$-$t$ path. Therefore, \cref{lemma:typei} will turn out to be sufficient for type I candidates.

\subsubsection{Proof of \cref{lemma:concat}} \label{sec:concat}
Finally, we prove \cref{lemma:concat}. 
\concat*
\begin{proof} [Proof of \cref{lemma:concat}]
    We will only show that there is an algorithm $\textsc{ConcatX}(P, c)$ in $\tilde O(\abs{P} + n ^ 2)$ time and space that can compute an array $\pi[*, *]$ where for any pair of vertices $(s, t) \in \ver ^ 2$, $\pi[s, t]$ is either an $s$-$t$ path with a hop of $O(h)$ in $G$ or $\perp$, such that for all $p_0, p_1 \in P$ such that $c \in p_0$ and $p_0 \circ p_1$ is an $s$-$t$ path, $w(\pi[s, t]) \le w(p_0 \circ p_1)$. We can similarly show that there is an algorithm $\textsc{ConcatY}(P, c)$ in $\tilde O(\abs{P} + n ^ 2)$ time and space that can compute an array $\pi[*, *]$ where for any pair of vertices $(s, t) \in \ver ^ 2$, $\pi[s, t]$ is either an $s$-$t$ path with a hop of $O(h)$ in $G$ or $\perp$, such that for all $p_0, p_1 \in P$ such that $c \in p_1$ and $p_0 \circ p_1$ is an $s$-$t$ path, $w(\pi[s, t]) \le w(p_0 \circ p_1)$. Therefore by taking the element-wise minimum, we get the desired results for $\textsc{Concat}(P, c)$.
    
    The implementation for $\textsc{ConcatX}(P, c)$ is described in \cref{alg:concat}. The algorithm clearly runs in $\tilde O(\abs{P} + n ^ 2)$ time and space. We now show that in \cref{alg:concat}, for any pair of vertices $(s, t) \in \ver ^ 2$, $\pi[s, t]$ is either an $s$-$t$ path with a hop of $O(h)$ in $G$ or $\perp$, such that for all $p_0, p_1 \in P$ such that $c \in p_0$ and $p_0 \circ p_1$ is an $s$-$t$ path, $w(\pi[s, t]) < w(p_0 \circ p_1)$.

    It is easy to see that $\pi[s, t]$ is either $\perp$ or an $s$-$t$ path not longer than the shortest possible concatenation of the following three parts:
    \begin{itemize}
        \item A $s$-$c$ path $p_{\textrm{from}}$ such that for some $p \in P$, $p_{\textrm{from}}$ is a prefix of $p$,
        \item A $c$-$u$ path $p_{\textrm{to}}$ such that for some $p ^ {\prime} \in P$, $p_{\textrm{to}}$ is a suffix of $p ^ {\prime}$, and
        \item A $u$-$t$ path $p_{\textrm{other}} \in P$.
    \end{itemize}
    Since paths in $P$ have a hop of at most $h$, $\pi[s, t]$ has a hop of at most $3h = O(h)$. Let $p_0, p_1 \in P$ be such that $c \in p_0$ and $p_0 \circ p_1$ is an $s$-$t$ path. Then when $u = p_0[\abs{p_0}] = p_1[0]$, $p_{\textrm{from}} = {p_0}[<c]$, $p_{\textrm{to}} = {p_0}[>c]$, and $p_{\textrm{other}} = p_1$, $p_{\textrm{from}} \circ p_{\textrm{to}} \circ p_{\textrm{other}} = p \circ q$. Therefore, $w(\pi[s, t]) \le w(p_{\textrm{from}} \circ p_{\textrm{to}} \circ p_{\textrm{other}}) = w(p \circ q)$.
\end{proof}
\begin{algorithm} [H]
    \caption{Algorithmic Proof of \cref{lemma:concat}} \label{alg:concat}
    \begin{algorithmic}[1]
        \Procedure{\textsc{concatX}}{$P, c$}
            \State $P_{\textrm{from}}[*] \gets \perp$
            \State $P_{\textrm{to}}[*] \gets \perp$
            \State $\pi[*, *] \gets \perp$
            \For {$p \in P$}
                \If {$w(p) < w(\pi[0, p[\abs{p}]])$}
                    \State $\pi[0, p[\abs{p}]] \gets p$
                \EndIf
                \If {$c \in p$}
                    \If {$w(p[<c]) < w(P_{\textrm{from}}[p[0]])$} \label{line:getpc}
                        \State $P_{\textrm{from}}[p[0]] \gets p[<c]$
                    \EndIf
                    \If {$w(p[>c]) < w(P_{\textrm{to}}[p[\abs{p}]])$}
                        \State $P_{\textrm{to}}[p[\abs{p}]] \gets p[>c]$
                    \EndIf
                \EndIf
            \EndFor
            \For {$(u, t) \in \ver ^ 2$}
                \State $p \gets P_{\textrm{to}}[u] \circ \pi[u, t]$
                \If {$(w(p) < w(\pi[c, t])$}
                    \State $\pi[c, t] \gets p$
                \EndIf
            \EndFor
            \For {$(s, t) \in \ver ^ 2$}
                \State $p \gets P_{\textrm{from}}[s] \circ \pi[c, t]$
                \If {$w(p) < w(\pi[s, t])$}
                    \State $\pi[s, t] \gets p$
                \EndIf
            \EndFor
        \EndProcedure
    \end{algorithmic}
\end{algorithm}

\subsection{Final Update Algorithm With Type II Candidates} \label{sec:typeii}
\subsubsection{Why A Simple Approach Would Not Work}
It might be tempting to employ an approach similar to \cref{alg:basicupdate} for type II candidates. Specifically, after recovering type I paths using \cref{alg:updateconcatenation}, we first let $\Pi_o$ be a copy of $\overline{\Pi ^ {\prime}}$, and then in increasing order of $j$, we go through all $(s, x, t, j)$ with $X[s, x, t, j] = 1$ where both the path $\Pi_o[s, x, j - 1]$ and the path $\Pi_o[x, t, j - 1]$ are available, and recover $\Pi_o[s, t, j]$ using these two paths.

The issue with this approach is that we can no longer relate the running time of the algorithm to the running time of \cref{alg:oracleupdate}. Consider a secondary pair $(s, t)$ where the final shortest path has a large hop. It could be the case that a lot of low-hop paths can be used to recovered low-hop $s$-$t$ paths, but none of these paths are optimal. A lot of computation is wasted on computing low-hop paths which requires larger hitting sets, whereas the in \cref{alg:oracleupdate}, we only need a small hitting set for a large hop. We can see that this will not become an issue if whenever a path is recovered, the recovered path is guaranteed to be optimal. We will enforce this condition using a Dijkstra-like path extension procedure.
\subsubsection{Algorithm Description}
For type II candidates, we consider recovering paths the other way around. Whenever a secondary path from $s$ to $t$ with hop level $j$ gets recovered, we perform an \emph{extension} on the pair $(s, t)$. To do this, we enumerate all other pairs that use the recovered $s$-$t$ path in their recovery, which are $\{(x, t) \mid \sum_{j ^ {\prime} \in [j, i_h]}X[x, s, t, j ^ {\prime}] \ge 1\}$, for which we combine the path from $x$ to $s$ and the path from $s$ to $t$, and $\{(s, y) \mid \sum_{j ^ {\prime} \in [j, i_h]}X[s, t, y, j ^ {\prime}] \ge 1\}$, for which we combine the path from $s$ to $t$ and the path from $t$ to $y$. The core idea here is that at a fixed hop level $j$ the extension cost will be of the same order as the high probability recovery cost, since with high probability $\sum_{x \in \ver, j ^ {\prime} \in [j, i_h]}X[x, s, t, j ^ {\prime}]$ and $\sum_{y \in \ver, j ^ {\prime} \in [j, i_h]}{X[s, t, y, j ^ {\prime}]}$ are both $\sum_{j ^ {\prime} \in [j, i_h]}O(n / h_{j ^ {\prime}}) = O(n / h_j)$, which is of the same order as the recovery cost $c(s, t, j) = n / h_j$. Therefore, if we can make sure that we only perform the extension exactly once for every secondary path, with high probability, the total running time for these extensions will be of the same order as the total update cost of secondary paths, which from \cref{lemma:oracleupdatecost} is $\tilde O(n ^ {2.5})$. To ensure that a secondary path only gets extended once, we recover paths in a similar fashion as Dijkstra's algorithm: we always extend the shortest unextended secondary path. The final update algorithm dealing with both types of candidates is described in \cref{alg:update}.
\begin{algorithm} [H]
    \caption{Final Data Structure: Update} \label{alg:update}
    \begin{algorithmic}[1]
        \Function{\textsc{Update}}{}
            \State Compute $\Pi ^ {\prime}$ using the lines \ref{line:oraclegetpiprimebegin} to \ref{line:oraclegetpiprimeend} from \ref{alg:oracleupdate}.
            \State $\overline{\Pi ^ {\prime}} \gets \textsc{UpdateForConcatenations}()$
            \State $Q = \{\}$
            \For {$(s, t) \in \ver ^ 2$} \label{line:startofforloop}
                \If {$w({\overline{\Pi ^ {\prime}}}[s, t, i_h]) < w({\Pi ^ {\prime}}[s, t])$} \label{line:if4}
                    \State $\Pi_o[s, t] \gets \overline{\Pi ^ {\prime}}[s, t, i_h]$ \label{line:getpathtypei}
                    \State $Q \gets Q \cup \{(s, t)\}$
                \Else
                    \State $\Pi_o[s, t] \gets {\Pi} ^ {\prime}[s, t]$ \label{line:getpath}
                \EndIf
            \EndFor
            \State $A = \textsc{RandGetShortestPaths}(G, H)$
            \While {$Q$ is not empty}
                \State $(s, t) \gets \argmin(\{w(\Pi_o[s, t]) \mid (s, t) \in Q\})$ \label{line:findargmin}
                \If {$w(\Pi_o[s, t]) \le A_{s, t} \cap \abs{\Pi_o[s, t]} \le H$} \label{line:if1}
                    \State $Q \gets Q \backslash \{(s, t)\}$
                    \State $j \gets h ^ {-1}(\abs{\Pi_o[s, t]})$ \label{line:getj}
                    \For {$x$ such that $\sum_{j ^ {\prime} \in [j, i_h]}X[x, s, t, j ^ {\prime}] \ge 1$} \label{line:forx}
                        \State $p \gets \Pi_o[x, s] \circ \Pi_o[s, t]$
                        \If {$w(p) < w(\Pi_o[x, t])$} \label{line:if2}
                            \State $\Pi_o[x, t] \gets p$ \label{line:getpatha}
                            \State $Q \gets Q \cup \{(x, t)\}$
                        \EndIf
                    \EndFor
                    \For {$y$ such that $\sum_{j ^ {\prime} \in [j, i_h]}X[s, t, y, j ^ {\prime}] \ge 1$} \label{line:fory}
                        \State $p \gets \Pi_o[s, t] \circ \Pi_o[t, y]$
                        \If {$w(p) < w(\Pi_o[s, y])$} \label{line:if3}
                            \State $\Pi_o[s, y] \gets p$ \label{line:getpathb}
                            \State $Q \gets Q \cup \{(s, y)\}$
                        \EndIf
                    \EndFor
                \EndIf
            \EndWhile
            \For {$s, t \in \ver$}
                \State $A_{s, t} \gets \min(A_{s, t}, w(\Pi_o[s, t]))$
            \EndFor
            \State \Return $A$
        \EndFunction
    \end{algorithmic}
\end{algorithm}
In order to implement the set $Q$ in $\tilde O(1)$ time per operation and $\tilde O(\abs{Q})$ space, one can use a Fibonacci Heap similar to how it is commonly employed in the Dijkstra's Algorithm \cite{dijkstra, tarjanheap}.

\subsubsection{Correctness}
The next lemma proves the correctness of \cref{alg:update}.
\begin{lemma} \label{lemma:finalcorrectness}
    With high probability, at the end of \cref{alg:update}, for every $(s, t) \in \ver ^ 2$ such that the shortest $s$-$t$ path has a hop at most $H$, ${\Pi_o}[s, t]$ is the shortest $s$-$t$ path, and for every $(s, t) \in \ver ^ 2$, $A_{s, t}$ is the length of the shortest $s$-$t$ path.
\end{lemma}
\begin{proof}
    We do induction on the hop of the shortest paths. Let $(s, t) \in \ver ^ 2$ be such that the shortest $s$-$t$ path has a hop of $h ^ {\prime} \le H$. Suppose that for every other $(s ^ {\prime}, t ^ {\prime}) \in \ver ^ 2$ such that the shortest $s ^ {\prime}$-$t ^ {\prime}$ path has a hop less than $h ^ {\prime}$, ${\Pi_o}[s ^ {\prime}, t ^ {\prime}]$ is the shortest $s ^ {\prime}$-$t ^ {\prime}$ path at the end of \cref{alg:update}.

    Let $j = h ^ {-1}(h ^ {\prime})$. From the same argument as the one used in the proof of \cref{lemma:inductionstep}, with high probability there is some candidate $v \notin (s, t)$ such that $X[s, v, t, j] = 1$ and the concatenation of the shortest $s$-$v$ path and the shortest $v$-$t$ path is the shortest $s$-$t$ path.  

    If $v$ is a type I candidate, from \cref{lemma:typei}, $\overline{\Pi ^ {\prime}}[s, t, i_h]$ is the shortest $s$-$t$ path and $\Pi_o[s, t]$ will be assigned this path on line \ref{line:getpathtypei}. It is easy to see that $\Pi_o[s, t]$ will remain the shortest $s$-$t$ path for the rest of the procedure.

    If $v$ is a type II candidate, then one of $(s, v)$ and $(v, t)$ is secondary. Since both the shortest $s$-$v$ path and the shortest $v$-$t$ path have a hop smaller than $h ^ {\prime}$, by assumption, $\Pi_o[s, v]$ gets assigned the shortest $s$-$v$ path and $\Pi_o[v, t]$ gets assigned the shortest $v$-$t$ path both at some point of the procedure. Without loss of generality suppose that $\Pi_o[s, v]$ gets assigned the shortest $s$-$v$ path after $\Pi_o[v, t]$ gets assigned the shortest $v$-$t$ path. Since one of $(s, v)$ and $(v, t)$ is secondary and a primary path always gets assigned the shortest path on line \ref{line:getpath} before any secondary paths are recovered, $(s, v)$ must be secondary. Therefore $\Pi ^ {\prime}[s, v]$ is not the shortest $s$-$v$ path, and therefore $\Pi_o[s, v]$ must have been assigned the shortest $s$-$v$ path on line \ref{line:getpathtypei}, \ref{line:getpatha} or \ref{line:getpathb}, any of which would have added to $(s, v)$ to $Q$. Note that from \cref{lemma:long} we know that $A_{s, v}$ is no less than the length of the shortest $s$-$v$ path. Since the shortest $s$-$v$ path has a hop smaller than $h ^ {\prime} \le H$, it passes the condition check on line \ref{line:if1} and therefore $(s, v)$ must have been extended at some point in the algorithm. When $(s, v)$ was extended, since $X[s, v, t, j] = 1$ and $h ^ {-1}(\abs{\Pi_o[s, v]}) \le j$, we must have used $\Pi_o[s, v] \circ \Pi_o[v, t]$ to update $\Pi_o[s, t]$ inside the for loop starting on line \ref{line:fory}. Since we assumed that $\Pi_o[s, v]$ becomes the shortest path after $\Pi_o[v, t]$. When we extend $(s, v)$, both $\Pi_o[s, v]$ and $\Pi_o[v, t]$ are already the shortest paths. Thus, $\Pi_o[s, t]$ will be assigned the shortest $s$-$t$ path on line \ref{line:getpathb}. It is easy to see that $\Pi_o[s, t]$ will remain the shortest $s$-$t$ path for the rest of the procedure.
     
    Therefore, for every $(s, t) \in \ver ^ 2$ such that the shortest $s$-$t$ path has a hop at most $H$, ${\Pi_o}[s, t]$ is the shortest $s$-$t$ path. From a similar argument as the one in the proof of \cref{lemma:basiccorrectness}, with high probability, for every $(s, t) \in \ver ^ 2$, $A_{s, t}$ is the length of the shortest $s$-$t$ path.
\end{proof}

\subsubsection{Running Time Analysis}
The next three lemmata limit the running time of \cref{alg:update}.
\begin{lemma} \label{lemma:passsecondary}
    With high probability, if a pair of vertices $(s, t) \in \ver ^ 2$ passes the condition check on line \ref{line:if1}, $(s, t)$ is a secondary pair of vertices.
\end{lemma}
\begin{proof}
    Suppose the shortest path for a pair of vertices $(s, t) \in \ver ^ 2$ has a hop more than $H$. From \cref{lemma:long}, with high probability, $A_{s, t}$ is no greater than the length of the shortest path. Thus it is impossible for $(s, t)$ to pass the condition check on line \ref{line:if1}. Thus with high probability, if $(s, t)$ the condition check on line \ref{line:if1}, the shortest $s$-$t$ path has a hop at most $H$, and therefore $(s, t)$ is primary or secondary.

    If $(s, t)$ is primary, then ${\Pi} ^ {\prime}[s, t]$ is already the shortest path. Therefore, the condition on line \ref{line:if4} must be false and $\Pi_o[s, t]$ gets assigned ${\Pi} ^ {\prime}[s, t]$ on line \ref{line:getpath}. Since $\Pi_o[s, t]$ is already the shortest, the conditions on line \ref{line:if2} and \ref{line:if3} will never be satisfied for this pair and therefore $(s, t)$ will never be added to $Q$. Therefore $(s, t)$ will not even enter line \ref{line:if1}.
\end{proof}

\begin{lemma} \label{lemma:noduplicate}
    Any pair of vertices $(s, t) \in \ver ^ 2$ is extracted from $Q$ at most once.
\end{lemma}
\begin{proof}
    If a pair of vertices $(s, t) \in \ver ^ 2$ is extracted from $Q$, $w(\Pi_o[s, t])$ is the shortest among all $(s, t) \in Q$. Since the edge weights are nonnegative, all future additions to $Q$ will be some pair of vertices $(s ^ {\prime}, t ^ {\prime})$ with $w(\Pi_o[s ^ {\prime}, t ^ {\prime}]) \ge w(\Pi_o[s, t])$. However, one needs an improvement on $\Pi_o[s, t]$ to add $(s, t)$ to $Q$ again, which is not possible. Therefore, any pair of vertices $(s, t) \in \ver ^ 2$ is extracted from $Q$ at most once.
\end{proof}

\begin{lemma}
    \Cref{alg:update} runs in $\tilde O(n ^ {2.5})$ time with high probability.
\end{lemma}
\begin{proof}
    If the condition on \ref{line:if1} is satisfied for a pair of vertices $(s, t) \in \ver ^ 2$, from \cref{lemma:passsecondary}, $(s, t)$ is a secondary pair. $(s, t)$ will only be extended once due to \cref{lemma:noduplicate}. Note that since the algorithm is correct with high probability due to \cref{lemma:finalcorrectness}, when $(s, t)$ gets extended, $\Pi_o[s, t]$ is already the shortest path with high probability. Thus with high probability, on line \ref{line:getj}, $j$ becomes the hop level of the shortest $s$-$t$ path. Since the number of candidates for $x$ and $y$ on line \ref{line:forx} and on line \ref{line:fory} are both $\sum_{j ^ {\prime} \in [j, i_h]}O(n / h_{j ^ {\prime}}) = O(n / h_j)$ with high probability, the running time for its extension will be $\tilde O(n / h_j)$ with high probability. Since the update cost $c(s, t, j) = n / h_j$, with high probability, the total running time for extension is no more than the total update cost of the secondary paths up to a logarithmic factor, which from \cref{lemma:oracleupdatecost} is $\tilde O(n ^ {2.5})$.

    It is easy to verify that the rest of \cref{alg:update} runs in $\tilde O(n ^ {2.5})$ time. Therefore, the running time of \cref{alg:update} is $\tilde O(n ^ {2.5})$ with high probability.
\end{proof}

\section{Space Efficiency and Tie-breaking}
\subsection{Space Efficiency} \label{sec:space}
There are two major obstacles to achieving space efficiency. Firstly, we need to implement the random variables $X[*, *, *, *]$ within $\tilde O(n ^ 2)$ space while still being able to implement all operations involving these random variables efficiently. Secondly, we need to efficiently store the paths involved in our algorithm while still being able meet the requirements in \cref{sec:po}.

To achieve space efficiency, we have to implement the random variables $X[*, *, *, *]$ within $\tilde O(n ^ 2)$ space while still being able to implement all operations involving these random variables efficiently. To do this, instead of keeping the random variables constant throughout the update procedure, every time we access the random variable $X[s, v, t, j]$ for $(s, v, t, j) \in \ver ^ 3 \times [0, i_h]$, we resample the 0-1 random variable to be $1$ with probability $\frac{100\log_2{n}}{h_j} = \tilde O(1 / h_j)$. This way we do not actually need to store the variables themselves.

At the end of the update procedure, the random variables take values sampled at different times. To prove the correctness, at each update, for each $(s, v, t, j) \in \ver ^ 3 \times [0, i_h]$, we take a \emph{snapshot} of $X[s, v, t, j]$ that is equal to the value of this random value at some time in the update procedure. The way to take the snapshot is as follows:
\begin{itemize}
    \item If the hop level of the shortest $s$-$t$ path is not $j$, or if the shortest $s$-$t$ path is not the concatenation of the shortest $s$-$v$ path and the shortest $v$-$t$ path, take the snapshot arbitrarily.
    \item If both $(s, v)$ and $(v, t)$ are primary (i.e. $v$ is a type I candidate), let the snapshot be the value sampled on line \ref{line:recoveryhittingset} of \cref{alg:updateconcatenation}.
    \item If at least one of $(s, v)$ and $(v, t)$ is secondary (i.e. $v$ is a type II candidate), if in \cref{alg:update}, $\Pi_o[s, v]$ gets assigned the shortest $s$-$v$ path after $\Pi_o[v, t]$ gets assigned the shortest $v$-$t$ path, let the snapshot be the value sampled on line \ref{line:fory} in \cref{alg:update}, otherwise let the snapshot be the value sampled on line \ref{line:forx} in \cref{alg:update}.
\end{itemize}
It is easy to see that the execution of our update algorithm is correct as long as it would be correct if random variables were pre-sampled to be equal to their snapshots at the start of the execution and remained the same throughout the execution. Since the snapshots have the same distribution as the random variables themselves, and the pre-sample algorithm has been shown to be correct, the new sample-on-the-spot algorithm must also be correct.

Clearly, this new implementation can be done in $\tilde O(n ^ 2)$ space as we do not need to store any information. To tackle the operations involving these random variables efficiently, we note that there are two types of operations that we need to implement:
\begin{itemize}
    \item Type I: given some $(s, t, j) \in \ver ^ 2 \times [0, i_h]$, efficiently enumerate all $v$ such that $X[s, v, t, j] = 1$. (e.g. line \ref{line:recoveryhittingset} of \cref{alg:updateconcatenation}).
    \item Type II: given some $(s, t, j) \in \ver ^ 2 \times [0, i_h]$, efficiently enumerate all $v$ such that $\sum_{j ^ {\prime} \in [j, i_h]}X[s, v, t, j ^ {\prime}] \ge 1$. (e.g. line \ref{line:fory} of \cref{alg:update}, with parameters in a different order).
\end{itemize}
Since we now resample the random variables on the spot, Type I operations require us to randomly sample a subset of $\ver$ where each vertex is sampled with probability $\frac{100\log_2{n}}{h_j}$, and Type II operations require us to randomly sample a subset of $\ver$ where each vertex is sampled with probability $1 - \prod_{j ^ {\prime} \in [j, i_h]}(1 - \frac{100\log_2{n}}{h_{j ^ {\prime}}})$. It is well-known that both operations can be done efficiently with a running time proportional to the size of the output subset (e.g. by sampling the subset size from a Poisson distribution and then performing a sampling without replacement).

To efficiently store the paths involved in our algorithm while still being able to access the nodes in order efficiently as required in \cref{sec:po}, we can see that all the paths involved in our main algorithm are either of the following:
\begin{enumerate}
    \item A path on the rooted tree storing the outputs of \cref{alg:sshdp}.
    \item A concatenation of two other paths.
\end{enumerate}
To store the first type of path we simply store the pointers to the starting and ending nodes on the rooted tree. To access the nodes efficiently, we simply traverse the path on the tree. To efficiently store and access the second type of path, for each path, if one of the paths being concatenated contains no edges, the path is simply the other path being concatenated; otherwise, we simply store two pointers to the two paths that are concatenated. Since we never point to a path with no edges, we can easily access the nodes on the second type of path in linear time by recursion. Finally, to find the length and the hop of a path of the first type, we can use standard techniques on rooted trees. One easy way to do this as follows: Observe that the starting and ending nodes for paths involved in the data structure are always such that one is an ancestor of the other. Therefore, both the hop and the length are simply the difference between the hops and the lengths of two paths starting from the root. We can simply store the hops of lengths of paths from the root to all nodes on the tree in $O(n)$ space and answer queries in $O(1)$ time. To be able to find the length and the hop of a path of the second type in $O(1)$ time, we can simply use $O(1)$ space per path to store the hop and length of each path, and every time we concatenate two paths in the data structure, we compute the hop and the length of the new path as the sum of the hops and lengths of the two paths being concatenated.

The other parts of our data structure can clearly be implemented using $\tilde O(n ^ 2)$ space.

\subsection{Tie-Breaking}\label{sec:breaktie}
In this paper, we have assumed that for every $h$, the $h$-hop shortest paths are unique. Previously in \cite{demetrescu04}, a way of achieving this uniqueness even in the deterministic case without loss of efficiency for shortest paths is introduced by using ``extended weights.'' Since our algorithm is randomized, we will use a more straightforward randomized method based on random perturbation. The perturbation is introduced in the following way: for each edge $e$ of weight $w$, we change its weight to $$w \times n ^ {10} + n ^ {9} + \lambda(e),$$ where $\lambda(e)$ is a uniformly random integer in $[0, n ^ 8)$. We call a preprocessing ``major'' if the entire data structure is preprocessed from the top layer. Our algorithm's major preprocessing steps are $\Theta(n ^ {0.5})$ updates apart. At each major preprocessing, we re-introduce random perturbation on every edge present on the graph. The perturbation is erased at the next major preprocessing before a new perturbation is introduced. 

It is easy to see that a hop-restricted shortest path in the new graph must be a corresponding hop-restricted shortest path in the original graph. We will show that with high probability the hop-restricted shortest paths after the perturbation are unique. For every $(s, t, h) \in \ver ^ 2 \times [0, n - 1]$, let $d[s, t, h]$ be the length of the shortest $s$-$t$ path with a hop of exactly $h$. We have the following:
\begin{lemma} \label{lemma:uniqueprec}
    For any given graph $G$, if we perturb the edge weights in the way above, with a probability of $1 - n ^ {-3}$, for every $(s, t, h) \in \ver ^ 2 \times [1, n - 1]$, there is only one vertex $t ^ {prime}$ such that $d[s, t ^ {\prime}, h - 1] + w(t ^ {\prime}, t) = d[s, t, h]$.
\end{lemma}
\begin{proof}
    Given $(s, t, h) \in \ver ^ 2 \times [1, n - 1]$, note that from our perturbation, the values $\{(d[s, t ^ {\prime}, h - 1] + w(t ^ {\prime}, t)) \mod n ^ 8 \mid t \in \ver \backslash t ^ {\prime}\}$ are $n - 1$ uniformly random integers in $[0, n ^ 8)$, and by the union bound the probability that these $n - 1$ random integers are distinct is at least $1 - 1 / n ^ 6$. If these values are distinct, the values $\{d[s, t ^ {\prime}, h - 1] + w(t ^ {\prime}, t) \mid t \in \ver \backslash t ^ {\prime}\}$ must also be distinct, and therefore there is only one minimum.

    The claim follows from the union bound on all $n ^ 2(n - 1) \le n ^ 3$ tuples.
\end{proof}

Now we prove that with high probability the hop-restricted shortest paths are unique.
\begin{lemma}
    For any given graph $G$, if we perturb the edge weights in the manner above, with a probability of $1 - n ^ {-3}$, for every $(s, t, h) \in \ver ^ 2 \times [1, n - 1]$, the shortest $h$-hop-restricted $s$-$t$ path is unique.
\end{lemma}
\begin{proof}
    We argue that as long as the condition in \cref{lemma:uniqueprec} holds, for every $(s, t, h)$, the shortest $h$-hop-restricted $s$-$t$ path is unique.  

    Suppose for some $(s, t, h) \in \ver ^ 2 \times [1, n - 1]$ there are two distinct shortest $h$-hop-restricted $s$-$t$ paths $p_0$ and $p_1$. Due to the way we perturb the edge weights, $\abs{p_0} = \floor{(w(p_0) \mod n ^ {10}) / n ^ 9} = \floor{(w(p_1) \mod n ^ {10}) / n ^ 9} = \abs{p_1}$. Suppose $i \in [1, \abs{p_0} - 1]$ is the largest index such that $p_0[i] \ne p_1[i]$. Let $t ^ {\prime} = p_0[i + 1] = p_1[i + 1]$. We have $d[s, t ^ {\prime}, i + 1] = w(p_0[0, i + 1]) = d[s, p_0[i], i] + w(p_0[i], t ^ {\prime})$ and $d[s, t ^ {\prime}, i + 1] = w(p_1[0, i + 1]) = d[s, p_1[i], i] + w(p_1[i], t ^ {\prime})$, which is a contradiction to the condition in \cref{lemma:uniqueprec}.

    Since \cref{lemma:uniqueprec} holds with probability $1 - n ^ {-3}$, with probability $1 - n ^ {-3}$, for every $(s, t, h) \in \ver ^ 2 \times [0, n - 1]$, the shortest $h$-hop-restricted $s$-$t$ path is unique.
\end{proof}

\begin{lemma}
    The uniqueness condition holds for all graphs and sub-graphs involved in the data structure between two major preprocessing steps with probability at least $1 - 1 / n ^ {1.49}$.
\end{lemma}
\begin{proof}
    Between the major preprocessing steps, there are $O(n ^ {0.5})$ updates, and for each update, we preprocess the $O(\log n)$ layers, and therefore there are $\tilde O(n ^ {0.5})$ preprocessing steps. Each preprocessing involves $\tilde O(n)$ sub-graphs for which we need to ensure that the uniqueness condition holds. Therefore by the union bound, the uniqueness condition holds for all these graphs with probability at least $1 - 1 / n ^ {1.49}$.
\end{proof}

\bibliographystyle{alphaurl} 
\bibliography{main}

\end{document}